%% file: main.tex
\documentclass[journal,10pt,twocolumn,oneside]{IEEEtran}

\usepackage{graphicx}
\usepackage{subfigure}
\usepackage{epstopdf}
\usepackage{booktabs}
\usepackage{wrapfig}
\usepackage{lipsum}
\usepackage{verbatim}
\usepackage{algorithm}
\usepackage{algorithmicx}
\usepackage{algpseudocode}
\usepackage{makecell,rotating,multirow,diagbox}%关于表格制作
\usepackage[cmex10]{amsmath}
\usepackage{amsthm}

\usepackage{nicefrac}

\def\0{{\mathbf 0}}
\def\1{{\mathbf 1}}
\def\s{{\mathbf s}}
\def\t{{\mathbf t}}
\def\v{{\mathbf v}}
\def\x{{\mathbf x}}
\def\y{{\mathbf y}}
\def\z{{\mathbf z}}

\def\c{{\mathbf c}}

\def\A{{\mathbf A}}
\def\B{{\mathbf B}}

\def\D{{\mathbf D}}

\def\L{{\mathbf L}}
\def\N{{\mathbf N}}
\def\A{{\mathbf A}}
\def\P{{\mathbf P}}
\def\Q{{\mathbf Q}}

\def\X{{\mathbf X}}
\def\Y{{\mathbf Y}}

\def\Z{{\mathbf Z}}

\def\W{{\mathbf W}}
\def\e{{\mathbf e}}
\def\I{{\mathbf I}}

\def\cG{{\mathcal G}}

\def\cL{{\mathcal L}}
\def\cS{{\mathcal S}}

\def\cV{{\mathcal V}}
\def\cE{{\mathcal E}}
\def\cO{{\mathcal O}}
\def\cR{{\mathcal R}}

\def\bphi{{\pmb{\phi}}}
\def\bpsi{{\pmb{\psi}}}

\def\bphi{{\pmb{\phi}}}
\def\bpsi{{\pmb{\psi}}}
\usepackage{color}
\usepackage{mdwmath}
\usepackage{url}
\usepackage{amssymb}%安装数学字体的包
\DeclareMathOperator*{\argmax}{argmax}
\newtheorem{lemma}{\textbf{Lemma}}
\newtheorem{theorem}{\textbf{Theorem}}

\newtheorem{proposition}{\textbf{Proposition}}

\newtheorem{corollary}{\textbf{Corollary}}
\newcommand{\red}[1] {\textcolor[rgb]{1.0,0.0,0.0}{{#1}}}

\begin{document}
\title{Graph Sampling for Matrix Completion Using Recurrent Gershgorin Disc Shift}
\author{%\IEEEauthorblockN{abc}
\IEEEauthorblockN{Fen Wang, Yongchao Wang, \emph{Member, IEEE}, {Gene Cheung, \emph{Senior Member, IEEE} and Cheng Yang, \emph{Member, IEEE}}}
\renewcommand{\baselinestretch}{1.0}
\thanks{F. Wang conducted this work during her visit to York University under the scholarship from China Scholarship Council. \emph{(Corresponding author: Yongchao Wang and Gene Cheung.)}}
\thanks{F. Wang and Y. Wang are with State Key Laboratory of ISN, Xidian University, Xi'an 710071, Shaanxi, China (e-mail: fenwang@stu.xidian.edu.cn; ychwang@mail.xidian.edu.cn).}
\thanks{{G. Cheung and C. Yang are with the department of EECS, York University, 4700 Keele Street, Toronto, M3J 1P3, Canada (e-mail:genec@yorku.ca; cyang@eecs.yorku.ca).}}
}
\maketitle
%
% The abstract is a short summary of the work to be presented in the article.
\begin{abstract}
\input{abstract}
\end{abstract}
%
% Keywords. The author(s) should pick words that accurately describe the work being
% presented. Separate the keywords with commas.

%\begin{teaserfigure}
%  \includegraphics[width=\textwidth]{sampleteaser}
%  \caption{Seattle Mariners at Spring Training, 2010.}
%  \Description{Enjoying the baseball game from the third-base seats. Ichiro Suzuki preparing to bat.}
%  \label{fig:teaser}
%\end{teaserfigure}

%
% This command processes the author and affiliation and title information and builds
% the first part of the formatted document.
\section{Introduction}
\input{intro}

\graphicspath{{figures/}}
\section{Related works}
\label{sec:related}
\input{related}

\section{Problem Formulation}
\label{sec:formulate}
\input{formulate}

\section{Fast Sampling on Product Graph via Gershgorin Circle Theorem}
\label{sec:sampling1}
\input{sampling1}

\section{Iterative Graph Spectral Sampling for Matrix Completion}
\label{sec:sampling2}
\input{sampling2}

% \section{Graph Sampling Based on Dual Graph Bandlimitedness}
% \label{sec:bandlimited}
% \input{bandlimitedSampling}

\section{Graph Construction}
\label{sec:graph-construction}
\input{graph}

\section{Experimentation}
\label{sec:results}
\input{results}

\section{Conclusion}
\label{sec:conclude}
\input{conclude}

\section*{Acknowledgement}
The authors would like to thank the great help from Wes Eardley for investigating available datasets for simulations.

\appendices

\input{appen}

\bibliographystyle{IEEEtran}
\bibliography{main}

\end{document}

%% file: abstract.tex
Matrix completion algorithms fill missing entries in a large matrix given a subset of observed samples.
However, how to best pre-select informative matrix entries given a sampling budget is largely unaddressed. 
In this paper, we propose a fast sample selection strategy for matrix completion from a graph signal processing perspective. 
Specifically, we first regularize the matrix reconstruction objective using a dual graph signal smoothness prior, resulting in a system of linear equations for solution.
We then select appropriate samples to maximize the smallest eigenvalue $\lambda_{\min}$ of the coefficient matrix, thus maximizing the stability of the linear system. 
To efficiently solve this combinatorial problem, we derive a greedy sampling strategy, leveraging on Gershgorin circle theorem, that iteratively selects one sample (equivalent to shifting one Gershgorin disc) at a time corresponding to the largest magnitude entry in the first eigenvector of a modified graph Laplacian matrix.  
Our algorithm benefits computationally from warm start as the first eigenvectors of incremented Laplacian matrices are computed recurrently for more samples. 
To achieve computation scalability when sampling large matrices, we further rewrite the  coefficient matrix as a sum of two separate components, each of which exhibits block-diagonal structure that we exploit for alternating block-wise sampling. 
Extensive experiments on both synthetic and real-world datasets show that our graph sampling algorithm substantially outperforms existing sampling schemes for matrix completion and reduces the completion error, when combined with a range of modern matrix completion algorithms. 

%% file: intro.tex
Big data means not only that the volume of acquired data is large, but the dimensionality of the dataset is also considerable. 
\textit{Matrix completion} (MC) \cite{liu2017MC} is an example of  this ``curse of dimensionality" problem, where two large dimensional item sets (\textit{e.g.}, viewers and movies in the famed Netflix challenge) are correlated within and across sets. 
Specifically, given a small subset of pairwise observations (\textit{e.g.}, viewers' ratings on movies), an MC algorithm reconstructs missing entries in the target matrix signal. 
Many MC algorithms have been devised using different priors to regularize the under-determined inverse problem, such as low rank of the target matrix \cite{candes2009exact} and graph signal smoothness priors \cite{EPFLgraph}. 
See \cite{candes10} for an introductory exposition.

While {MC} has been investigated intensively, how to pre-select matrix entries to collect informative samples given a sampling budget is largely unaddressed.
This sampling problem is of practical concern for applications where sampling is expensive and/or time-consuming \cite{rubens2015active}  (\textit{e.g.}, requesting viewers to fill out movie surveys is cumbersome and costly). 
% In \cite{2013activeMC}, three different active querying strategies were proposed based on the reconstruction uncertainty of each entry, which was further investigated in \cite{huang2018active}. 
% Adaptive sensing was proposed in \cite{adaptive-sampling} to select a subset of informative columns for matrix completion with bounded complexity. 
Conventionally, entries in matrix were collected based on their informative uncertainty computed using different methods \cite{2013activeMC,adaptive-sampling,huang2018active}, which are typically computation-intensive. 
There exist fast sampling strategies for coherent matrices that assigned each entry a probability for non-uniform random sampling \cite{LLS}. 
However, performance of random  selection schemes is in general inferior compared to their deterministic counterparts. 

Recently, rating matrix in a recommendation system was investigated from a \emph{graph signal processing} (GSP) perspective \cite{ortega18ieee}, where the target matrix signal was assumed to be bandlimited / smooth with respect to both the row (movies) and column (viewers) graphs (called  \textit{factor graphs}), \cite{huang2018rating,EPFLgraph}. 
Under this assumption, \cite{PG,tensor} identified a \emph{structured} set for MC {by sampling entries that are intersections of greedily selected rows and columns.}
%The reason authors in \cite{PG,tensor} proposed this \emph{structured sampling} philosophy is that sampling on the product graph directly is unrealistic for the explosive dimensionality. 
However, the imposed structure severely limits the possible sampling patterns and thus is too restrictive to achieve a general sampling budget. 
More general sampling methods for single graphs are not applicable for MC due to their high complexities  on the \textit{product graph} (one large graph containing all matrix entries as nodes) \cite{anis16}. 

In contrast, in this paper we propose a fast \emph{unstructured} graph sampling method for MC. 
We first regularize the sampling objective with a \emph{dual graph smoothness prior}---a generalization of the well-known Tikhonov regularizer \cite{bishop1995training} to the graph signal domain---which was shown effective in completing missing matrix entries previously \cite{EPFLgraph,monti2017geometric}. 
This formulation leads to a system of linear equations for solution, which can be computed efficiently using known numerical linear algebra algorithms such as \textit{conjugate gradient} (CG) \cite{CG}. 
To maximize the stability of the linear system, we select samples to maximize the smallest eigenvalue $\lambda_{\min}$ of the coefficient matrix 
\footnote{Maximizing $\lambda_{\min}$ of a matrix is also known as the  \textit{E-optimality} criterion in optimal design of experiments  \cite{chen15sampling,experiments,Boyd}.}, 
which we show to also mean minimizing the upper bound of the reconstructed matrix signal's squared error. 

We propose to optimize the formulated objective greedily: select one node at a time such that the current sample set results in the largest $\lambda_{\min}$. 
However, in each greedy step, computing $\lambda_{\min}$ for all candidates and choosing the largest one would still be expensive.
Instead, leveraging on an insightful corollary of the Gershgorin circle theorem \cite{yuji}, we greedily select the sample corresponding to the largest magnitude entry in the first eigenvector of an augmented coefficient matrix, which also minimizes a related objective.  
%To compute the first eigenvector fast for large matrices, we use the locally optimal block preconditioned conjugated gradient (LOBPCG) method which benefits computationally from warm start with good initialization inherited from last greedy step. 
Our algorithm benefits from warm start as the first eigenvectors of incrementally updated Laplacian matrices are computed recurrently during sampling using the well-known \textit{locally optimal block preconditioned conjugated gradient} (LOBPCG) method \cite{lobpcg}. 

To achieve computation scalability when sampling large matrices, we further partition the coefficient matrix into two matrices, each exhibiting attractive block-diagonal structure after permutation. 
We then propose an iterative sampling strategy that efficiently collects a pre-determined number of samples block-wise on smaller blocks alternately. 
Extensive experiments on synthetic and real-world datasets show that our proposed graph sampling methods  achieve much smaller RMSE than competing sampling schemes for MC  \cite{PG,akie2018eigenFREE,puy18}, when combined with a variety of state-of-the-art MC methods \cite{GRALS,berg2017graph,NMC}. 

The outline of the paper is as follows.
We first overview related works in graph sampling and active matrix completion in Section \ref{sec:related}.
We then derive our graph sampling objective for MC using the dual graph smoothness prior in Section \ref{sec:formulate}. 
In Section \ref{sec:sampling1}, we describe our sampling strategy via the Gershgorin circle theorem, and then we  propose an iterative block-wise sampling scheme for large matrices in Section \ref{sec:sampling2}. 
Finally, extensive experiments and conclusion are presented in Section \ref{sec:results} and \ref{sec:conclude}, respectively. 

%% file: related.tex
%This paper proposes to solve the active matrix completion problem (selecting matrix entries actively) from graph sampling perspective, by assuming the matrix signal is smooth on dual graphs. Therefore, in this  section, 
We first discuss related works in graph sampling.
Then we review some literature in conventional active matrix completion domain. 

\subsection{Subset Sampling of Graph Signals }

Subset sampling of graph signals is a fundamental problem in GSP \cite{ortega18ieee}: how to select a node subset in a graph for sampling such that the remaining samples in the signal can be reconstructed with high accuracy. 
Most existing works \cite{pesenson08,anis14,chen15sampling,anis16,wang18,shomorony14,puy18,PG} extended the notion of critical sampling (also known as \textit{Nyquist sampling}) in regular data kernels to bandlimited / smooth signals on graphs, where \textit{graph frequencies} are defined as {eigenvalues} of a graph variation operator like the graph Laplacian or adjacency matrix. 
Sampling methods in GSP can be broadly divided into two categories: i) deterministic schemes \cite{pesenson08,anis14,chen15sampling,anis16,wang18,chamon2018}, and ii) random schemes \cite{shomorony14,puy18}. 

Most deterministic schemes
\cite{chen15sampling,MFN,chamon2018} assumed that the graph signal is \textit{bandlimited}: its spectral coefficients are concentrated on a set of extreme eigenvectors. 
%Those strategies exploit the greedy idea for collecting samples and require computation of the extreme eigenvectors for the supported eigenvector information and computation-intensive operations in each greedy step.
\cite{anis16} proposed a lightweight sampling method using the notion of \emph{spectral proxies}, which collected samples based on the first eigenvector of a submatrix in each greedy step.  %However, when the graph is very large, its computation via LOBPCG is hard to converge because of the degree $k$, thus leading to bad performance. 
One recent work \cite{wang18} avoided eigenvector computation via Neumann series expansion, but required a large number of matrix series multiplications for accurate approximation. 
Recently, \cite{yuanchao2019ICASSP} proposed a sampling method based on Gershgorin disc alignment without any explicit eigen-decompositions.  
\cite{akie2018eigenFREE} also proposed an eigen-decomposition-free sampling method based on localization operator's coverage surface,
{but had no notion of global errors in its optimization objective}. 
%, which is computation-expensive to approximate using the Chebyshev polynomials when graph is very large. 
However, those sampling methods cannot be directly applied to real-world MC problem because of  their high complexities on the corresponding product graph. 

In parallel, \cite{puy18} proposed a non-uniform random graph sampling scheme to select nodes  based on the notion of \textit{graph coherence}, such that each node was sampled with a designed probability.  
However, the performance of random sampling is generally inferior  compared to its deterministic competitors.

{To the best of our knowledge, we are the first to propose an unstructured and deterministic graph sampling strategy specifically for MC in the literature, with complexity roughly linear to the size of the {factor graphs}.  }
%\red{did u define what a factor graph is earlier? how bout a product graph? I don't think so. u can mention them in the intro, but here if u want to criticize previous GSP works on factor graphs, u have to properly define them first.}
%\blue{fen:dual graph means row and column graphs; factor graph means row or column graphs ( $m$ or $n$ nodes); product graph means the kronecker product of row and column graphs ( $m*n$ nodes)}

\subsection{{Active Matrix Completion}}
\label{subsec:relatedMC}

In the active learning literature, strategically selecting matrix entries is also called \emph{active matrix completion} \cite{activeMC2015}. 
%For instance, in recommendation systems, one specific user may be asked to rate some certain items to collect more informative entries in terms of recovery. 
Active matrix completion approaches can be categorized into two types: i) statistical approach \cite{2013activeMC,activeMC2015}, and ii) GSP approach \cite{PG,tensor}. 

% {$\bullet$ Active Learning Approach}
Among works pursuing a statistical approach, \cite{activeMC2015} formulated an active learning objective for MC, which was tackled using collaborative filtering. 
In \cite{2013activeMC}, three different active querying strategies were proposed based on the reconstruction uncertainty of each entry; this  querying idea was further investigated in  \cite{ruchansky2015matrix,mak2017active,huang2018active}. 
However, those methods are generally computation-expensive since they must evaluate all candidates based on expected reconstructed error. 
Separately, adaptive sensing was proposed in \cite{adaptive-sampling} to select a subset of informative columns for MC with bounded complexity. 
Using the coherent property of matrix signals, \cite{LLS} proposed a fast \textit{leveraged score-based sampling} (LSS) to assign each matrix entry a sampling probability for non-uniform random sampling.
Nevertheless, the performance of random sampling is not comparable to the deterministic strategies.    
%{$\bullet$ Graph Signal Processing Approach}

Recently, from a GSP viewpoint, the target matrix was interpreted as a bandlimited signal on the two factor graphs \cite{huang2018rating}. 
Based on such bandlimited model, \cite{PG,tensor} proposed an efficient \textit{structured} graph sampling strategy for MC, and then extended it to multidimensional tensor graph signals, whose effectiveness has been validated in recommendation system and point cloud sampling. 
However, structured sampling---selected samples must correspond to matrix entries that are intersections of chosen rows and columns of the matrix---is too restrictive to  achieve arbitrary sampling budgets. 
In this paper, we propose a fast \textit{unstructured} graph sampling strategy for MC, with comparable complexity to the structured counterpart \cite{PG,tensor}. 

%% file: formulate.tex
We derive an objective function for matrix sampling using a dual graph signal smoothness prior  \cite{EPFLgraph,monti2017geometric}. 
We first define the graph-based MC problem in Section \ref{subsec:matrixComplete}, and then formulate the graph spectral matrix sampling problem in Section \ref{subsec:formulateSampling}.

\subsection{Dual Graph Smoothness based Matrix Completion}
\label{subsec:matrixComplete}
Denote the original matrix signal and additive noise by $\X$ and $\N$ respectively, where $\X, \N \in \mathbb{R}^{m \times n}$.
Given a sampling set $\Omega=\{(i,j) \;|\; i \in\{1,\ldots,m\},\; j \in \{1,\ldots,n\}\}$, its corresponding sampling operator $\mathbf{A}_{\Omega} \in \{0, 1\}^{m \times n}$ can be defined as
\begin{equation}\label{sampling}
\begin{split}
{\mathbf{A}_{\Omega}(i,j)} = \left\{ \begin{array}{ll}
1,&\mbox{if}\; {(i,j) \in {\Omega}}; \\
0,&\mbox{otherwise}.
\end{array} \right.
\end{split}
\end{equation}

With the above notations, the sampled noise-corrupted observation is $\Y=\A_{\Omega}\circ(\X+\N) \in \mathbb{R}^{m\times n}$, where $\circ$ denotes the element-wise matrix multiplication operator.
We assume that elements in noise $\mathbf{N}$ are zero-mean, independent and identically distributed (i.i.d.) noise with the same variance. 

MC methods attempt to reconstruct the original matrix $\X$ from the partial noisy observations $\Y$, under an assumed prior for matrix $\X$, like low-rank \cite{candes2009exact}:
\begin{align}
&    \min_{\X}~~\textrm{rank}(\X)\\
&    \textrm{s.t.}~\|\A_{\Omega}\circ\X-\Y\|_F<\sigma\nonumber
\end{align}    
where $\sigma$ is set sufficiently small to enforce similar reconstruction of the observed samples $\Y$ in signal $\X$. 

Recently, \cite{EPFLgraph} introduced a \textit{dual graph smoothness prior} to promote low rank matrix reconstruction.
%\red{does dual graph smoothness promote low rank? can u prove this?}\blue{fen: as discussed, I cannot prove this.}
Specifically, columns of $\X$ are assumed to be smooth with respect to an undirected weighted  \textit{row graph} $\cG_r=\{\cV_r,\cE_r,\W_r\}$ with vertices $\cV_r=\{1,\dots,m\}$ and edges $\cE_r\subseteq \cV_r\times \cV_r$. 
Weight matrix $\W_r$ specifies pairwise similarities among vertices in $\cG_r$. 
The combinatorial graph Laplacian matrix of row graph $\cG_r$ is $\L_r=\D_r-\W_r$, where the degree matrix $\D_r$ is a diagonal matrix with entries $\D_r(i,i)=\sum_j\W_r(i,j)$. 
Taking the $j$-th column of $\X$, denoted by $\x_j$, as an example, the total graph variation of $\x_j$ on graph $\cG_r$ is defined as \cite{shuman13}: 
\begin{equation}\label{smoothness}
    \x^{\top}_j\L_r\x_j=\sum_{(k,l)\in\cE_r}\W_r(k,l)(\x_j(k)-\x_j(l))^2.
\end{equation}
Thus a smaller variation value would mean similar sample reconstructions between strongly connected nodes. 

Similarly, the rows of $\X$ are assumed smooth with respect to a \textit{column graph} $\cG_c=\{\cV_c,\cE_c,\W_c\}$ with vertices
$\cV_c=\{1,\dots,n\}$, edges $\cE_c \subseteq \cV_c \times \cV_c$ and weight matrix $\W_c$. 
Corresponding graph Laplacian matrix for the column graph $\cG_c$ is $\L_c=\D_c-\W_c$. 
Using the movie recommendation systems as an example, the row graph is a \emph{similarity graph} among movies, and the column graph is a \emph{social relationship graph} among viewers. 
Row and column graphs can be constructed from observed data using different methods \cite{EPFLgraph,huang2018matrix,hu2019feature}; we describe our adopted graph construction schemes in Section\;\ref{sec:graph-construction}.

We now formulate the MC problem with \textit{dual graph Laplacian regularization} (DGLR) \cite{jiahao2017graph} as follows:
\begin{align}\label{eq:obj0}
&\hspace{0cm}\min_{\mathbf{X}}~~f(\mathbf{X})=\frac{1}{2}\|\mathbf{A}_{\Omega}\circ(\mathbf{X}-\mathbf{Y})\|^{2}_{F}\\
&\hspace{1.8cm}+\frac{\alpha}{2}\text{Tr}\left(\mathbf{X}^{\top}\mathbf{L}_{r}\mathbf{X}\right)
+\frac{\beta}{2}\text{Tr}\left(\mathbf{X}\mathbf{L}_{c}\mathbf{X}^\top\right),\nonumber
\end{align}
where $\alpha$ and $\beta$ are parameters trading off the first fidelity term with the two signal smoothness priors. 

% \red{this is not necessary.} 
% Specifically, the dual graph smoothness actually penalizes
% \begin{align}
% \text{Tr}\left(\mathbf{X}^{\top}\mathbf{L}_{r}\mathbf{X}\right)=\sum^n_{j=1} \x^{\top}_j\L_r\x_j;\\ \text{Tr}\left(\mathbf{X}\mathbf{L}_{c}\mathbf{X}^\top\right)=\sum^m_{i=1}\bar{\x}^{\top}_i\L_c\bar{\x}_i,
% \end{align}
% where $\bar{\x}_i$ is the $i$-th row of matrix $\X$. This actually pushes the matrix signal to be jointly smooth in row and column graphs.

%Though \eqref{eq:obj0} does not explicitly include low rankness prior, 
It has been shown through extensive experiments that the dual graph signal smoothness prior enables good MC performance \cite{EPFLgraph}. 
More generally, a graph smoothness prior is a generalization of the well-known \textit{Tikhonov regularizer}---popular regularization for ill-posed problems---to graph data kernels \cite{ortega18ieee}. 
%\red{how come u r uysing numbers like 25 instead of actual citation??}
In the fast growing field of GSP \cite{shuman13}, 
%\red{once u have defined GSP, stop redefining it in the sequel!!} 
the graph smoothness prior has already been shown effective empirically for a wide range of inverse problems (\textit{e.g.}, image denoising / deblurring \cite{bai19tip,jiahao2017graph}, point cloud denoising \cite{zeng20183d,dinesh20183d}), and recently is successfully applied to MC also \cite{huang2018rating,EPFLgraph, monti2017geometric, tensor}. 
%\red{I can't believe u r using numbers instead of actual latex citations!! why??}
%See \blue{Appendix \ref{low-rank}} of the supplementary material for an exposition. 
Next, we derive our sampling algorithm based on this well-accepted prior in the GSP community.

To solve the unconstrained QP problem \eqref{eq:obj0}, we take the derivative of $f(\X)$ with respect to $\X$, set it to $0$ and solve for $\X$, resulting in a system of linear equations for unknown $\text{vec}(\X^*)$: 
\begin{equation}\label{eq:linEqleft}
    \left(\tilde{\mathbf{A}}_{\Omega}+\alpha\mathbf{I}_n\otimes\mathbf{L}_{r}
+\beta\mathbf{L}_{c}\otimes\mathbf{I}_m\right)\text{vec}(\mathbf{X}^{*})=\text{vec}(\mathbf{Y})
\end{equation}
where $\tilde{\mathbf{A}}_{\Omega}=\text{diag}(\text{vec}({\mathbf{A}}_{\Omega}))$, $\textrm{vec}(\cdot)$ means a vector form of a matrix by stacking its columns, and $\textrm{diag}(\cdot)$ creates a diagonal matrix with input vector as its diagonal elements.  
See {Appendix \ref{solution}} for a detailed derivation. 

Since the coefficient matrix $\Q=\tilde{\mathbf{A}}_{\Omega}+\alpha\mathbf{I}_n\otimes\mathbf{L}_{r}
+\beta\mathbf{L}_{c}\otimes\mathbf{I}_m$ is in general symmetric, sparse and positive definite (PD)\footnote{See {Appendix \ref{Q}} for a detailed description of $\Q$.}, \eqref{eq:linEqleft} can be solved efficiently using a plethora of mature numerical linear algebra methods such as \textit{conjugate gradient} (CG)  \cite{conjugatedGradient}. 
This is one notable appeal of formulating the MC problem using the dual graph signal smoothness prior in \eqref{eq:obj0}, where computing its solution requires only solving a system of linear equations.

%This result is actually an extension of one dimensional optimal solution in smoothness assumption, which is$\mathbf{x}=(\mathbf{D}^\top\mathbf{D}+\mu\mathbf{L})^{-1}\mathbf{D}^\top\mathbf{y}$. Here $\mathbf{D}^\top\mathbf{D}\backsim\tilde{\mathbf{A}}_{\Omega}$ and $\mathbf{D}^\top\mathbf{y}\backsim \text{Vec}(\mathbf{Y})$. 

\subsection{Graph Sampling for Matrix Completion  based on DGLR Formulation}
\label{subsec:formulateSampling}

%In this subsection, we will derive our graph sampling objective function. 
% Existing works in matrix completion \cite{EPFLgraph,monti2017geometric,berg2017graph} focus on the reconstruction of missing matrix entries given random collected training set. 
% Though of great importance, active sampling (also known as querying) of matrix entries is to some extent neglected. %Sampling is often expensive and/or time-consuming \cite{rubens2015active}, and pre-selection of the most beneficial samples given a sampling budget is an important problem. 
% \red{u keep saying the same thing over again. remove this!}
% We next derive a sampling objective based on \eqref{eq:linEqleft}.

The stability of the linear system in \eqref{eq:linEqleft} is determined by the \textit{condition number} of coefficient matrix $\Q$, which is the ratio of the largest eigenvalue $\lambda_{\max}$ of $\Q$ to its smallest eigenvalue $\lambda_{\min}$.
Given that $\lambda_{\max}(\Q)$ is upper-bounded for a degree-constrained graph (see {Appendix \ref{upper-bound}} for a proof), to maximize stability, we seek to maximize $\lambda_{\min}(\Q)$ through sampling, \textit{i.e.},

\begin{equation}\label{eq:obj1}
\mathop{\text{max}}\limits_{\Omega}~~g(\Omega)=\lambda_{\text{min}}\left(\tilde{\mathbf{A}}_{\Omega}+\alpha\mathbf{I}_n\otimes\mathbf{L}_{r}
+\beta\mathbf{L}_{c}\otimes\mathbf{I}_m\right)
\end{equation}

Maximizing $\lambda_{\min}$ of a coefficient matrix is also known as the  \textit{E-optimality} criterion in optimal design \cite{chen15sampling,experiments,Boyd}, and is a common objective for many well-known linear system optimizations, \textit{e.g.}, {active learning \cite{yu2006active}, {sensor placement} \cite{MPME2016TSP} and polynomial regression \cite{pukelsheim1993optimal}.} 
In our sampling scenario, we show further that maximizing \eqref{eq:obj1} also means minimizing the MSE upper bound, as stated formally in the lemma below. 
% See {Appendix \ref{proof-lemma}} for a proof.
% \red{we have a lot of Appendices. we can consider some of them to the main text, such as this one.}

\begin{lemma}
\label{lemma1} 

Given dual graph Laplacians $\L_r$ and $\L_c$, assuming ground truth signal $\X$ is corrupted by independent additive noise $\N$, MSE of the reconstructed signal $\mathbf{X}^{*}$ with respect to the original signal $\mathbf{X}$ is upper-bounded by 
\begin{equation}
\|\text{vec}(\mathbf{X}^{*})-\text{vec}(\mathbf{X})\|_2\leq\frac{\rho}{\lambda_{\min}(\mathbf{Q})}+
\|\text{vec}(\N)\|_2
\end{equation}
where $\rho=\|\left(\alpha\mathbf{I}_n\otimes\mathbf{L}_{r}+\beta\mathbf{L}_{c}\otimes\mathbf{I}_m\right)\left[\text{vec}(\mathbf{X+N})\right]\|_2$.
\end{lemma}
\begin{proof}
In vector form, 
$\text{vec}(\mathbf{Y})=\tilde{\mathbf{A}}_{\Omega}\left[\text{vec}(\mathbf{X+N})\right]$. 
Thus, the solution to the system of linear equations \eqref{eq:linEqleft} is 
\begin{align}\label{rec2}
& \text{vec}(\mathbf{X}^{*})=\mathbf{Q}^{-1}\textrm{vec}(\Y)=\mathbf{Q}^{-1}\tilde{\mathbf{A}}_{\Omega}\left[\text{vec}(\mathbf{X+N})\right]\\\nonumber
&=\mathbf{Q}^{-1}\left(\mathbf{Q}-\alpha\mathbf{I}_n\otimes\mathbf{L}_{r}
-\beta\mathbf{L}_{c}\otimes\mathbf{I}_m\right)\left[\text{vec}(\mathbf{X+N})\right] \\
& =\text{vec}(\mathbf{X})+\textrm{vec}(\N)\nonumber-
\mathbf{Q}^{-1}\L\left[\text{vec}(\mathbf{X+N})\right],\nonumber
\end{align}
where $\L=\alpha\I_n\otimes\L_r+\beta\L_c\otimes\I_m$.

Thus the squared error of estimator $\text{vec}(\X^{*})$ with respect to $\text{vec}(\X)$ is
\begin{align}\label{MSE}
& \|\text{vec}(\mathbf{X}^{*})-\text{vec}(\mathbf{X})\|_2\\\nonumber
&=\|\textrm{vec}(\N)-\mathbf{Q}^{-1}\L\text{vec}(\mathbf{X}+\N)\|_2\\\nonumber
&\leq\|\mathbf{Q}^{-1}\L\text{vec}(\mathbf{X}+\N)\|_2
+\|\textrm{vec}(\N)\|_2\\\nonumber
&\leq\|\mathbf{Q}^{-1}\|_2\|\L\text{vec}(\mathbf{X}+\N)\|_2
+\|\textrm{vec}(\N)\|_2\nonumber\\
&=\rho\|\mathbf{Q}^{-1}\|_2+\|\textrm{vec}(\N)\|_2\nonumber
\end{align}
where $\rho=\|\left(\alpha\mathbf{I}_n\otimes\mathbf{L}_{r}+\beta\mathbf{L}_{c}\otimes\mathbf{I}_m\right)\left[\text{vec}(\mathbf{X+N})\right]\|_2$. 

From inequality \eqref{MSE}, we see that sampling set $\Omega$ only influences the MSE upper bound by manipulating $\|\Q^{-1}\|_2$. 
Moreover, for symmetric and  positive definite matrix $\mathbf{Q}$, we know 
\begin{equation}\label{eq:norm}
\|\mathbf{Q}^{-1}\|_2=\lambda_{\text{max}}(\mathbf{Q}^{-1})
=\frac{1}{\lambda_{\text{min}}(\mathbf{Q})}.
\end{equation}
% \red{the $l_2$-norm of a matrix is its largest eigenvalue. see FSP textbook, for example. we don't need the first step.}
% \blue{fen: $l_2$ norm of one matrix is generally defined as its largest singular value; when the matrix is symmetric (of course square), it would be its largest eigenvalue. So I used the first step.}
We complete this proof by substituting \eqref{eq:norm} into equation \eqref{MSE}.
% Hence, we arrive at the following 
% \begin{equation}
% \|\text{vec}(\mathbf{X}^{*})-\text{vec}(\mathbf{X})\|_2\leq\frac{\rho}{\lambda_{\text{min}}(\mathbf{Q})}+
% \|\textrm{vec}(\N)\|_2
% \end{equation}
\end{proof}
In the next section, we will present a fast graph sampling strategy to solve optimization problem \eqref{eq:obj1}. 
%From Lemma \ref{lemma1}, we konw that maximizing $\lambda_{\min}$ of $\mathbf{Q}$ is actually to minimize the upper-bound of the MSE. In the next section, a specific sampling method will be proposed to solve problem \eqref{eq:obj1} with efficient implementation. 

%% file: sampling1.tex
For brevity, we interpret  $\L=\alpha\mathbf{I}_n\otimes\mathbf{L}_{r}+\beta\mathbf{L}_{c}\otimes\mathbf{I}_m$ as the Laplacian of a \emph{scaled product graph}\footnote{Cartesian product between two matrices $\L_r$ and $\L_c$ is defined by:
$\L_r \odot \L_c=\mathbf{I}_n\otimes\mathbf{L}_{r}+\mathbf{L}_{c}\otimes\mathbf{I}_m$.}.
Thus optimization \eqref{eq:obj1} becomes the maximization of $\lambda_{\min}$ for matrix $\tilde{\A}_{\Omega}+\L$.
By definition, $\tilde{\A}_{\Omega}$ is a diagonal matrix:
\begin{equation}\label{samplingProduct}
\begin{split}
\tilde{\A}_{\Omega}(l,l) = \left\{ \begin{array}{ll}
1,&\mbox{if}\; l\in\cS; \\
0,&\mbox{otherwise}.
\end{array} \right.
\end{split}
\end{equation}
where $    \cS=\left\{l|l=i+m\times(j-1),\forall(i,j)\in \Omega\right\}$.

Thus, coefficient matrix $\Q$ can be rewritten as:
\begin{equation}\label{eq:sumobj}
    \Q=\L+\tilde{\A}_{\Omega}
    =\L+\sum^{K}_{t=1}\e_{k_t}\e^{\top}_{k_t},
\end{equation}
where $K=|\cS|$, $k_t=\cS(t)$ and $\e_{k_t}$ is an indicator vector with $\e_{k_t}({k_t})=1$ and $\e_{k_t}(q)=0$ for $q\neq k_t$.

{Finding an optimal $\Omega$ (or $\cS$) to maximize $\lambda_{\min}(\Q)$ is combinatorial in nature.}
Towards a low-complexity sampling strategy, we take a greedy approach, where we iteratively add a locally optimal sample to a selected sample set until the sample budget is exhausted.
Hence, assuming we have collected $t-1$ samples in $\cS_{t-1}$, at the $t$-th iteration, we solve the following local optimization problem:
\begin{equation}\label{eq:greedyShfit}
   k^*_{t}= \argmax_{k_t\in\cS_{t-1}^c}~~\lambda_{\min}(\L_{t-1}+\e_{k_t}\e^{\top}_{k_t}),
\end{equation}
where $t\in \{1,\dots,K$\}, $\cS_{t}=\cS_{t-1}\cup k^*_t$ with  $\cS_0=\emptyset$, and
$\L_t=\L_{t-1}+\e_{k^*_t}\e^{\top}_{k^*_t}$ with $\L_0=\L$.

To find an optimal solution $k^*_t$ in \eqref{eq:greedyShfit} for each new sample, one can compute $\lambda_{\min}$ of the \emph{incremented Laplacian}\footnote{{We use ``increment'' here to mean increasing one diagonal element of a matrix by 1 (equivalently shifting the center of one Gershogrin disc right by 1), while other matrix entries remain unchanged.}} $\L_{t-1}+\e_{k_t}\e^{\top}_{k_t}$ corresponding to all candidate nodes $k_t \in \cS^c_{t-1}$ and identify the largest one, which is computation-intensive.
% \red{not sure ``shifted" is the right term. only one entry in the diagonal is incremented. ``shifted" often means $+ \epsilon \I$. maybe ``augmented" is more appropriate.}
% \blue{fen:please see the added footnote. I prefer to use the word ``shift'', since it means a move of  translation. ``Augmented'' is too general for this case.}
Instead, we circumvent multiple computations of the smallest eigenvalue for candidates using a strategy based on the \emph{Gershgorin circle theorem} (GCT).

\subsection{Gershgorin Disc Shift based Graph Sampling}
\label{gershgorin-sec}

We first review GCT and its corollary \cite{yuji}, which will lead to a lightweight sampling method later.

\begin{theorem}
\label{gershgorin}

Given an $n \times n$ matrix $\A$  with entries $a_{ij}$, define the $i$-th Gershgorin disc $D(a_{ii},R_i)$, corresponding to the $i$-th row of $\A$, with center $a_{ii}$ and radius
$R_i=\sum_{j\neq i}|a_{ij}|$.
Each eigenvalue $\lambda$ of $\A$ lies within at least one Gershgorin disc, \textit{i.e.},
\begin{equation}
\exists~ i ~~ | ~~
a_{ii}-R_i \leq \lambda \leq a_{ii}+R_i.
\end{equation}
\end{theorem}

\begin{comment}
\begin{proof}
We here write down a short proof for this well-known theorem for coherent understanding.
Assume an eigen-pair of $\A$ is $\{\lambda,\x\}$, then we know $\A\x=\lambda\x$.  If the $i$-th component of $\x$ has the largest absolute value, \textit{i.e.}, $|\x(i)|=\max_j~|\x(j)|$, then the $i$-th row of the eigen-pair linear equation is $\sum_{j\neq i}a_{ij}\x(j)+a_{ii}\x(i)=\lambda \x(i)$. This will induce
\begin{equation}
    |\lambda-a_{ii}|=\left|\sum_{j\neq i}a_{ij}\frac{\x(j)}{\x(i)}\right|\leq
    \sum_{j\neq i}|a_{ij}|
\end{equation}
\end{proof}
\end{comment}
\begin{corollary}\label{coro1}
If the largest magnitude component of an eigenvector $\x$ is at index $i$, then its corresponding eigenvalue $\lambda$ must be within the $i$-th Gershgorin disc $D(a_{ii},R_i)$.
\end{corollary}

This corollary implies that $\lambda_{\min}$ of matrix $\L_{t-1}$ must reside in the $j^*$-th Gershgorin disc, where $j^*=\argmax_j~|\bphi(j)|$ and $\bphi$ is the first eigenvector of $\L_{t-1}$ corresponding to $\lambda_{\min}$ \footnote{{In this paper, eigenvectors are all normalized, \textit{i.e.}, $\|\bphi\|_2=1$}.}.
By \eqref{eq:greedyShfit}, $\e_{k_t}\e^{\top}_{k_t}$
%would increase the diagonal element of $\L_{t-1}$ by 1 at location $\{k_t,k_t\}$ and keep all the rest elements unchanged, which in fact
shifts the center of the $k_t$-th Gershgorin disc of $\L_{t-1}$ to the right by 1.
\textit{
Our strategy
%\footnote{By restricting $k^*_t \in \mathcal{S}_{t-1}^c$, one cannot guarantee that $|\bphi(k^*_t)| \geq |\bphi(i)|, \forall i$, but in practice, an already chosen sample $j \in \mathcal{S}_{t-1}$ with self-loop at node $j$ in $\L_{t-1}$ is very unlikely to have the largest magnitude $|\bphi(j)|$.}
is then to right-shift the Gershgorin disc corresponding to the largest magnitude entry $k_t^* \in \mathcal{S}_{t-1}^c$ in $\bphi$ which contains $\lambda_{\min}$, thus promoting a larger $\lambda_{\min}$ in $\L_t$; \textit{i.e.}, select sample $k^*_t$ where
}
\begin{align}\label{eq:greedySelect}
&k^*_t=\mathop{\argmax}\limits_{k_t\in \cS^c_{t-1}}~~|\bphi(k_t)|, \\
&\textrm{s.t.} ~~\L_{t-1} \bphi = \lambda_{\min}(\L_{t-1})~ \bphi
\nonumber
\end{align}

\noindent
\textbf{Remark:}
To choose one sample, our strategy requires computation of only the first eigenvector of a sparse matrix $\L_{t-1}$ once, without multiple evaluations for all candidates.

%\red{check the wording change below.}

Note that in \eqref{eq:greedySelect} we select the index $k_t^*$ with the largest magnitude $|\bphi(k_t^*)|$ \textit{only} among entries in the unsampled set $\mathcal{S}_{t-1}^c$ instead of the entire vector, as specified in Corollary 1.
However, one can guarantee that the largest magnitude index $k_t^*$ in $\bphi$, in fact, only resides in $\mathcal{S}_{t-1}^c$, and thus \eqref{eq:greedySelect} is consistent with Corollary 1.
We state this formally in the following Proposition:

% \begin{proposition}\label{uniquesampling}
% $j^*=k^*_t$ even if $j^*$ is the optimal solution of problem \eqref{eq:greedySelect} by relaxing the search space from $\cS^c_{t-1}$ to $\cV$, i.e,  $j^*={\argmax}_{j\in \cV}~~|\bphi(j)|$, then $j^*=k^*_t$.
% \end{proposition}
%\red{this proposition is stated poorly.}

\begin{proposition}\label{uniquesampling}
$k^*_t$ computed from \eqref{eq:greedySelect} is also the index with the largest magnitude in $\bphi$, \textit{i.e.}, $k_t^*={\argmax}_{j\in \cV}~~|\bphi(j)|$.
\end{proposition}
\begin{proof}
Based on the definition below equation \eqref{eq:greedyShfit}, we can deduce that
\begin{equation}
\L_{t-1}=\L+\sum_{i\in\cS_{t-1}}\e_{i}\e^{\top}_{i}.
\end{equation}

Hence, in matrix $\L_{t-1}$, left-ends of Gershgorin discs corresponding to indices $i\in \cS_{t-1}$ are at 1, and the other discs' left-ends are at 0.
Suppose now that $j^*={\argmax}_{j\in \cV}~~|\bphi(j)|$ and  $j^*\in \cS_{t-1}$.
According to Corollary 1, the smallest eigenvalue $\lambda_{\min}(\L_{t-1})$ must be within $j^*$-th Gershgorin disc,  \textit{i.e.}, $\lambda_{\min}(\L_{t-1})\ge 1$.

% \blue{However, as proved in paper \cite{bai2019fast}, $\lambda_{\min}(\L_{t-1})$ is lower bounded by the smallest left-end and upper bounded by the largest left-end of Gershgorin discs of matrix $\L_{t-1}$, which are 0 and 1 respectively when  $\cS_{t-1}\subset \cV$. Therefore, $0\le\lambda_{\min}(\L_{t-1})\le1$. }

{We also know the first eigenvector of matrix $\L$ is a constant vector $\c=\frac{1}{\sqrt{mn}}[1,\dots,1]$ with eigenvaue 0. This yields:
\begin{align}
  & \lambda_{\min}(\L_{t-1})=\min_{\|\x\|_2=1}~\x^\top\left(\L+\sum_{i\in\cS_{t-1}}\e_{i}\e^{\top}_{i}\right)\x\nonumber\\
   &\le \c^\top\L\c+\c^\top\left(\sum_{i\in\cS_{t-1}}\e_{i}\e^{\top}_{i}\right)\c\\
   &=\|\c(\cS_{t-1})\|^2_2<1\nonumber
\end{align}
where the last inequality holds since $\cS_{t-1}\subset \cV$}.

% \blue{Because $\e_i\e^\top_i$ is PSD in nature and $\L_t$ is also PSD, from Weyl's inequality \cite{neumannseries}, we know that
% \begin{equation}
%     \lambda_{\min}(\L_t+\e_i\e^\top_i)\ge \lambda_{\min}(\L_t)+\lambda_{\min}(\e_i\e^\top_i)=\lambda_{\min}(\L_t),
% \end{equation}
% which implies $\lambda_{\min}(\L_t)$ is monotonically increasing with more samples.
% \red{I don't get the monoonicity argument.}
% When $\cS=\cV$, the incremented Laplacian is $\L+\I$ with $\lambda_{\min}=1$. Therefore, $0<\lambda_{\min}(\L_{t-1})<1$ for  $\cS_{t-1}\subset \cV$}.
% \red{There must be simpler ways to show $\lambda_{\min}(\L_{t-1}) < 1$. let's talk.}
This is contradictory to previous result that $\lambda_{\min}(\L_{t-1})\ge 1$.
%\red{how do we know $\lambda_{\min} < 1$ for $\mathcal{S}_{t-1} \subset \mathcal{V}$?}
Thus $j^*\in\cS^c_{t-1}$ and $j^*=k^*_t$.
\end{proof}

% \blue{
% \begin{theorem}
% Suppose $\B=\A+\c\c^\top$ where $\A$ is symmetric and $\c$ is non-zero vector, then,
% \begin{equation}
%     \lambda_1(\B)\ge  \lambda_1(\A) \ge \lambda_2(\B)\dots
%     \ge \lambda_n(\B)\ge \lambda_n(\A)
% \end{equation}
% \end{theorem}}
Thus, we conclude that entries within set $\cS_{t-1}$ cannot have the largest energy in first eigenvector $\bphi$ of $\L_{t-1}$.

We can alternatively justify our strategy by showing that the index chosen by our strategy optimizes a related objective to \eqref{eq:greedyShfit}.
We state this formally in the following lemma.
%The index to this particular Gershgorin disc is also the optimal solution of a \emph{proxy} of problem \eqref{eq:greedyShfit}, stated formally in the following lemma:

\begin{lemma}
\label{paper:lemma2}

An optimal solution to the problem
\begin{equation}\label{eq:proxy}
    \argmax_{k_t\in\cS_{t-1}^c}~~\lim_{\delta\rightarrow 0} ~ \lambda_{\min}(\L_{t-1}+\delta\e_{k_t}\e^{\top}_{k_t})
\end{equation}
is $k^*_t=\mathop{\argmax}\limits_{k_t\in \cS^c_{t-1}}~~|\bphi(k_t)|$, where $\bphi$ is the first eigenvector of $\L_{t-1}$ corresponding to smallest eigenvalue $\lambda_{\min}(\L_{t-1})$.
\end{lemma}

\vspace{0.1in}
\begin{proof}
Since matrix $\L_{t-1}$ is augmented by a small matrix $\delta\e_{k_t}\e^{\top}_{k_t}$, the resulting matrix is $
    {\tilde{\L}_t}=\L_{t-1}+\delta\e_{k_t}\e^{\top}_{k_t}
$,
where $k_t\in\cS_{t-1}^c$.
Using the Rayleigh quotient theorem \cite{neumannseries}, we can write $\lambda_{\min}$ of matrix $\tilde{\L}_t$ as
\begin{align}
 &   \lambda_{\min}(\tilde{\L}_t)=\min_{\x}~\frac{\x^{\top}\L_{t-1}\x+\delta\x^{\top}\e_{k_t}\e^{\top}_{k_t}\x}{\x^{\top}\x}\\
 &=\min_{\x}~\frac{\x^{\top}\L_{t-1}\x+\delta\x(k_t)^2}{\x^{\top}\x},\nonumber
\end{align}
where the {minimizer} $\x^*$ is the first eigenvector of $\L_{t-1}$ when $\delta\rightarrow 0$, \textit{i.e.}, $\lim_{\delta\rightarrow 0}\x^*=\bphi$.
Therefore,
\begin{align}
 &   \lim_{\delta \rightarrow 0}~\lambda_{\min}(\tilde{\L}_t)=\lambda_{\min}(\L_{t-1})+\delta{\bphi(k_t)^2},
\end{align}
where $\bphi^\top\bphi$ is omitted since $\|\bphi\|_2=1$.

% With previous collected $\cS_{t-1}$, $\lambda_{\min}({\L_{t-1}})$ is known as \textit{a prior},
% \red{what does this mean? The English makes no sense to me.}
{Given collected $\cS_{t-1}$, $\lambda_{\min}({\L_{t-1}})$ does not depend on $k_t$}.
Hence,
%an optimal sampling solution of  maximizing $\lambda_{\min}(\tilde{\L}_t)$ when $\delta\rightarrow 0$ is
\begin{align}
\label{eq:limit}
 &\argmax_{k_t\in\cS_{t-1}^c}~\lim_{\delta \rightarrow 0}~\lambda_{\min}(\tilde{\L}_t)=\argmax_{k_t\in\cS_{t-1}^c}~\delta{\bphi(k_t)^2}=k^*_t
\end{align}
\end{proof}
Thus, by computing $k^*_t$ using  \eqref{eq:greedySelect}, we are \textit{optimally} solving problem \eqref{eq:proxy}, which is a proxy approximating original \eqref{eq:greedyShfit}.

\subsection{Fast Repeated Eigenvector Computation with Warm Start}
{Our proposed formulation \eqref{eq:greedySelect}  requires computing the first eigenvector of $\L_{t-1}$ in each greedy step.
%which can be obtained via implicitly restarted Arnoldi method \cite{lehoucq1997arpack} and  Krylov-Schur algorithm \cite{stewart2002krylov}. \red{not sure u need to cite methods other than LOBPCG.}
In this paper, we will adopt the state-of-the-art LOBPCG method \cite{lobpcg} to compute the first eigenvector, which has been proved very efficient for large sparse matrices  \cite{anis16}. }
% Computing $k^*_t$ in \eqref{eq:greedySelect} is computationally light, since the first eigenvector $\bphi$ is computed only once to identify a new sample among candidates in $\mathcal{S}_{t-1}^c$.
% \red{u said this already. say instead how u would compute the first eigenvector in general.}
{With an initial input $\x_0$, in each iteration, LOBPCG works as follows:
\begin{enumerate}
    \item Multiply $\mathbf{L}_{t-1}$ with $\x_i\in \mathbb{R}^{mn}$  (guess of the first eigenvector in $i$-th iteration) with complexity $\cO(E_{t-1})$, where $E_{t-1}$ is the number of non-zero entries in  $\L_{t-1}$;
    \item Perform a Rayleigh-Ritz step \cite{lobpcg} to compute the combination coefficients, solving an eigenvalue problem with complexity $\cO(r^3)$. $r$ is the number of computed eigenvectors and in our problem, $r=1$;
    \item Update $\x_i$ based on  Rayleigh-Ritz coefficients, go to step (1) until convergence.
\end{enumerate}}

Our proposed algorithm can benefit computationally from warm start when deploying LOBPCG:
we use the estimated first eigenvector $\bphi$ of $\L_{t-1}$ in the last iteration as the initial guess $\x_0$ for $\L_t$.
The small change between $\L_{t}$ and $\L_{t-1}$ (the Forbenius norm difference is only 1) ensures a good initial guess, reducing the number of iterations for LOBPCG to converge.
{Simulation results in Section \ref{sec:results} show that warm start does reduce sampling time noticeably.}
% The \blue{effectiveness} of warm start will be demonstrated in the experimental section.
% \red{effect is probably the wrong word here.}
\input{algorithm1}
We write the pseudo code of our sampling strategy in Algorithm \ref{GCS}, called \emph{Gershgorin circle shift} (GCS)-based sampling.

%\red{make a complexity analysis section}
\subsection{Complexity Analysis}
The complexity of our proposed GCS method is dominated by three components:
i) $K$ times greedy search, ii) first eigenvector computation, and iii) finding the largest element's location in each greedy search.

Identifying the largest energy index in a vector with length $mn$ has complexity $\cO(mn)$.
The complexity of using LOBPCG to compute the first eigenvector of $\L_t$ is $\cO(E_tF_t)$, where $F_t$ is the number of iterations till convergence in LOBPCG.
%\red{u already $R$ to mean disc radius.}
Note that $E_0=E_1=\dots=E_{K-1}$ %\red{what's $E$??}\blue{fen:as introduced in subsection B, $E_{t-1}$ is the number of nonzero entries in matrix $\L_{t-1}$}
since $\L_t=\L_{t-1}+\e_{k^*_t}\e^{T}_{k^*_t}$, and the diagonal terms of $\L$ are all non-zero for a connected graph.
{Because $\L_0=\L=\alpha\I_n\otimes\L_r+\beta\L_c\otimes\I_m$, \textit{i.e.}, matrix $\L$ is consisted of $n$ matrix $\L_r$ and $m$ matrix $\L_c$, the number nonzero entries in $\L$ is at most  $E_0=\cO(n|\cE_r|+m|\cE_c|)$, where $|\cE_r|$ and $|\cE_c|$ are the numbers of edges in graph $\L_r$ and $\L_c$ respectively. }
%\red{I have no idea what u r talking about.}

{Therefore,
denoting by $F=\max\{F_0,\dots,F_{K-1}\}$, in each greedy step, the complexity of LOBPCG is $\cO((n|\cE_r|+m|\cE_c|)F)$;} combined with $K$ times greedy search and signal sorting, our GCS method has the complexity  $\cO(K(n|\cE_r|+m|\cE_c|)F+Kmn)$.
If the row graph and column graph are both sparse such that $|\cE_r|=\cO(m)$ and $|\cE_c|=\cO(n)$, then the complexity of GCS can be abbreviated as $\cO(KFmn)$.
Though the spectral proxy based sampling method in \cite{anis16} also computes the first eigenvector of a submatrix of $\L^p$ via LOBPCG, it does not benefit from warm start, and the complexity of computing $\L^p\x$  will be higher than $\L\x$ by at least by a factor $p$.

\subsection{{Explanation from Graph Spectral Energy Perspective}}
We now interpret the GCS sampling from an energy spreading perspective.
First, we define the absolute value of the first eigenvector of incremented Laplacian as \textit{graph spectral energy}.
Our proposed GCS method is to select node with the largest energy at each greedy step.
Once node $i$ is sampled, intuitively, the energy of this node and nodes near $i$ should decrease such that in the next step, the proposed strategy will not sample those nodes.
This is formally stated in the next lemma:
\begin{lemma}\label{lemma:bound}
Denote by $\lambda_0=\lambda_{\min}(\L_t)$, $\beta_0=\lambda_{\min}(\L_t+\e_i\e^\top_i)$
and
$\L_t\bphi=\lambda_0\bphi$ , $(\L_t+\e_i\e^\top_i)\bpsi=\beta_0\bpsi$. Then
\begin{equation}
    \lambda_0+\bpsi(i)^2\le\beta_0\le\lambda_0+\bphi(i)^2
\end{equation}
In our problem, $\L_t \in \{\L_0,\L_1,\dots,\L_{K-1}\}$.
\end{lemma}
\begin{proof}
From the Rayleigh quotient theorem \cite{neumannseries},
\begin{equation}\label{raylei1}
    \lambda_0=\min_{\|\x\|_2=1}~\x^\top\L_{t}\x=\bphi^\top\L_{t}\bphi.
\end{equation}
and,
\begin{align}\label{raylei2}
 &   \beta_0=\min_{\|\y\|_{2}=1}~\y^\top(\L_t+\e_i\e^\top_i)\y\\
 &=\bpsi^\top(\L_t+\e_i\e^\top_i)\bpsi=\bpsi^\top\L_{t}\bpsi+\bpsi(i)^2.\nonumber
\end{align}
From equation \eqref{raylei1}, we know  $\bpsi^\top\L_{t}\bpsi\ge\bphi^\top\L_{t}\bphi=\lambda_0$, which implies $\beta_0\ge\lambda_0+\bpsi(i)^2$.
From equation \eqref{raylei2}, we can derive that
$\beta_0\le\bphi^\top\L_{s}\bphi=\bphi^\top\L_{t}\bphi+\bphi(i)^2=\lambda_0+\bphi(i)^2$, which is exactly the right part of the lemma.
\end{proof}
% \red{what point do u want to get out of this lemma? it came out of no where. state what u want to say intuitively first, then state formally as a lemma.}
\begin{figure}
  \centering
  \includegraphics[width=100pt,height=80pt]{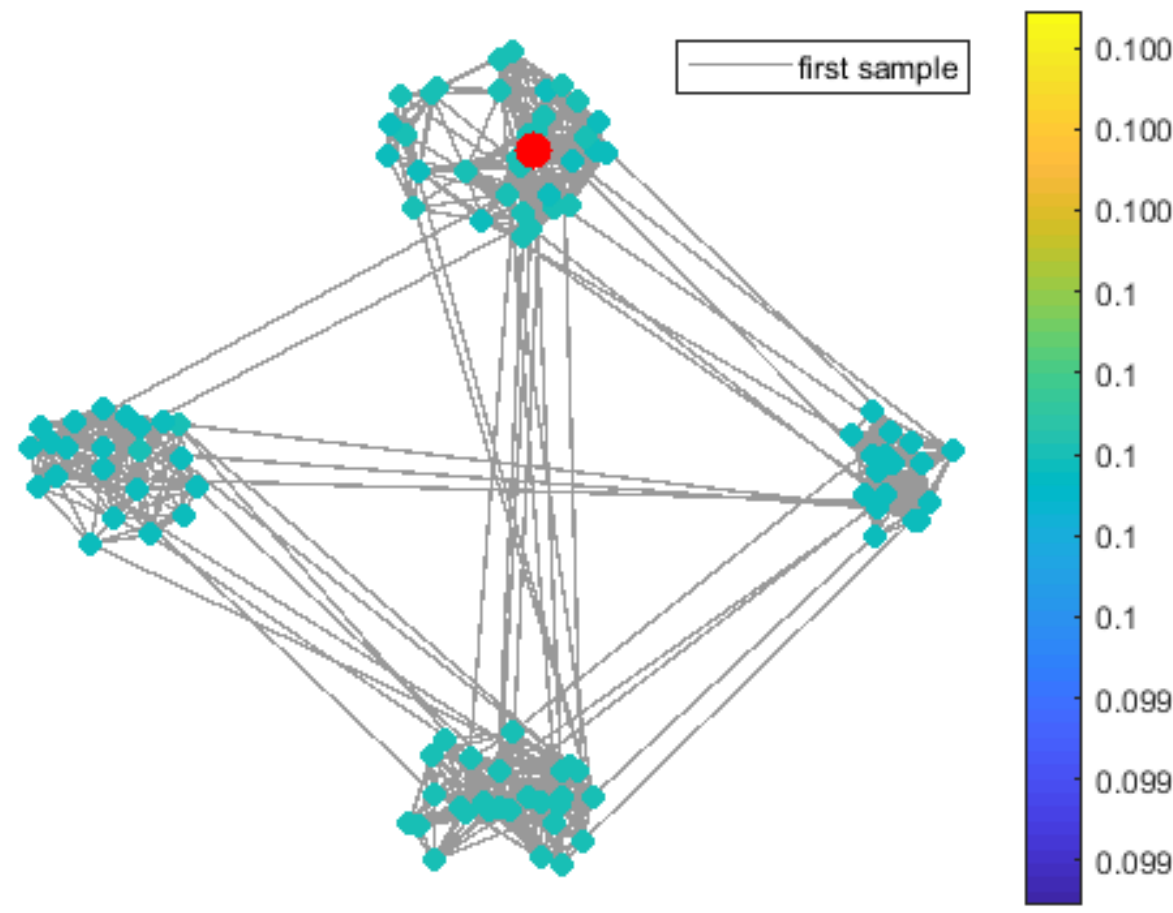}
  \includegraphics[width=100pt,height=80pt]{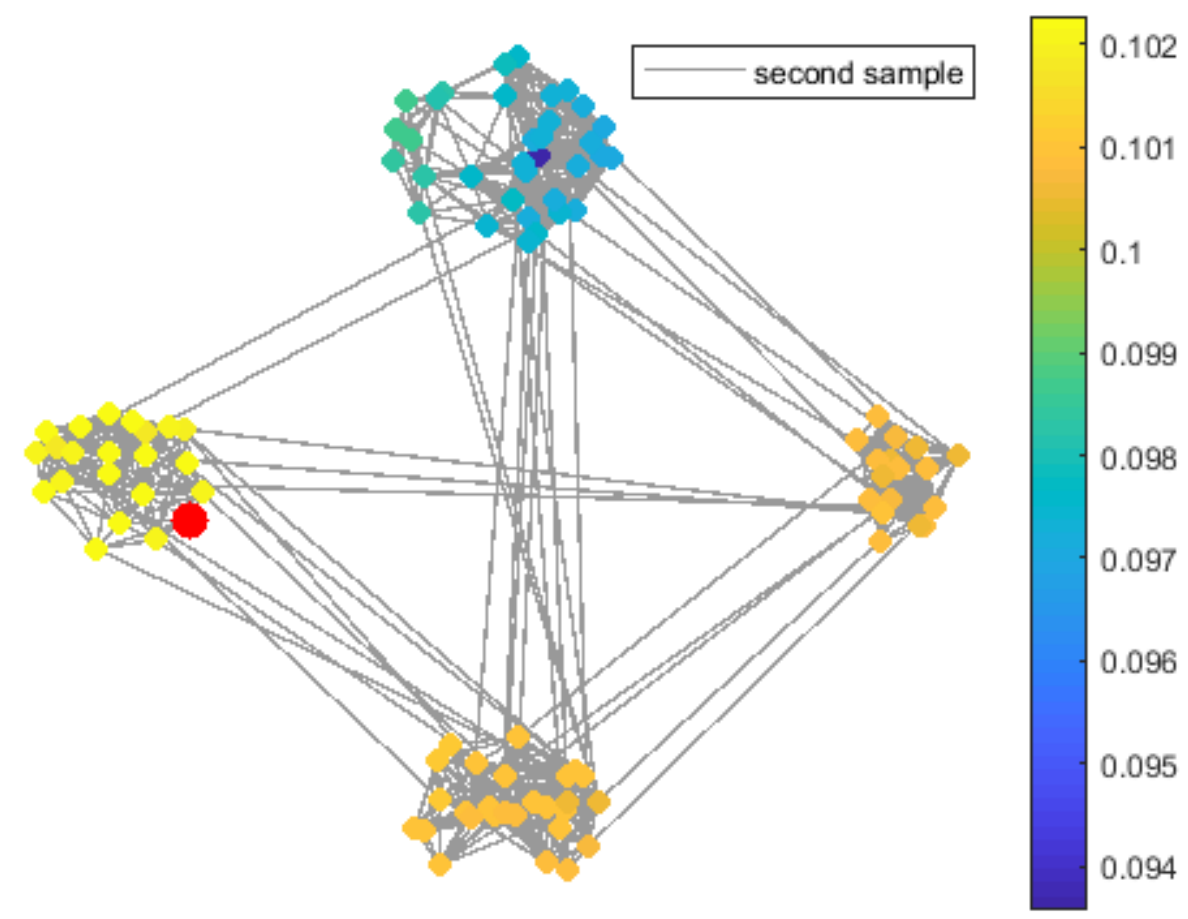}
  \includegraphics[width=100pt,height=80pt]{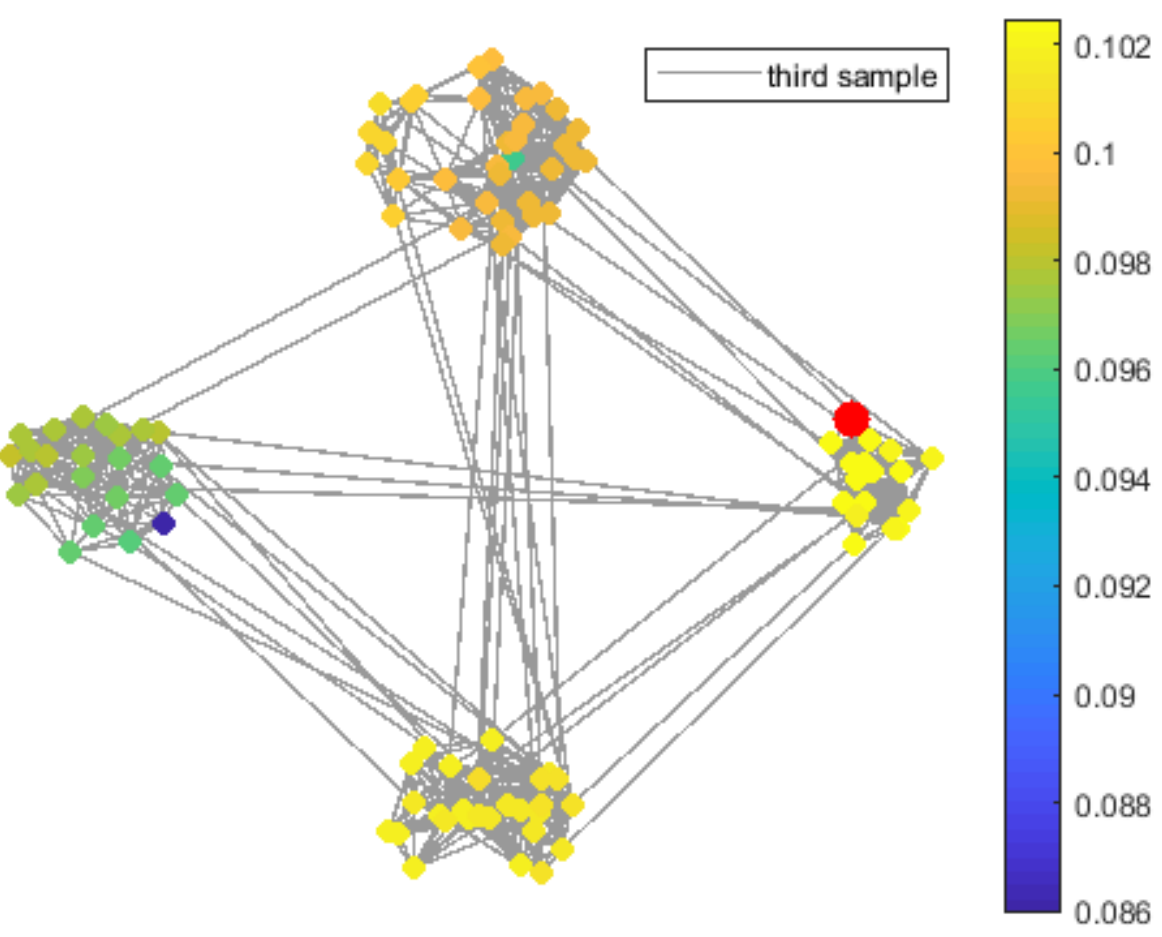}
  \includegraphics[width=100pt,height=80pt]{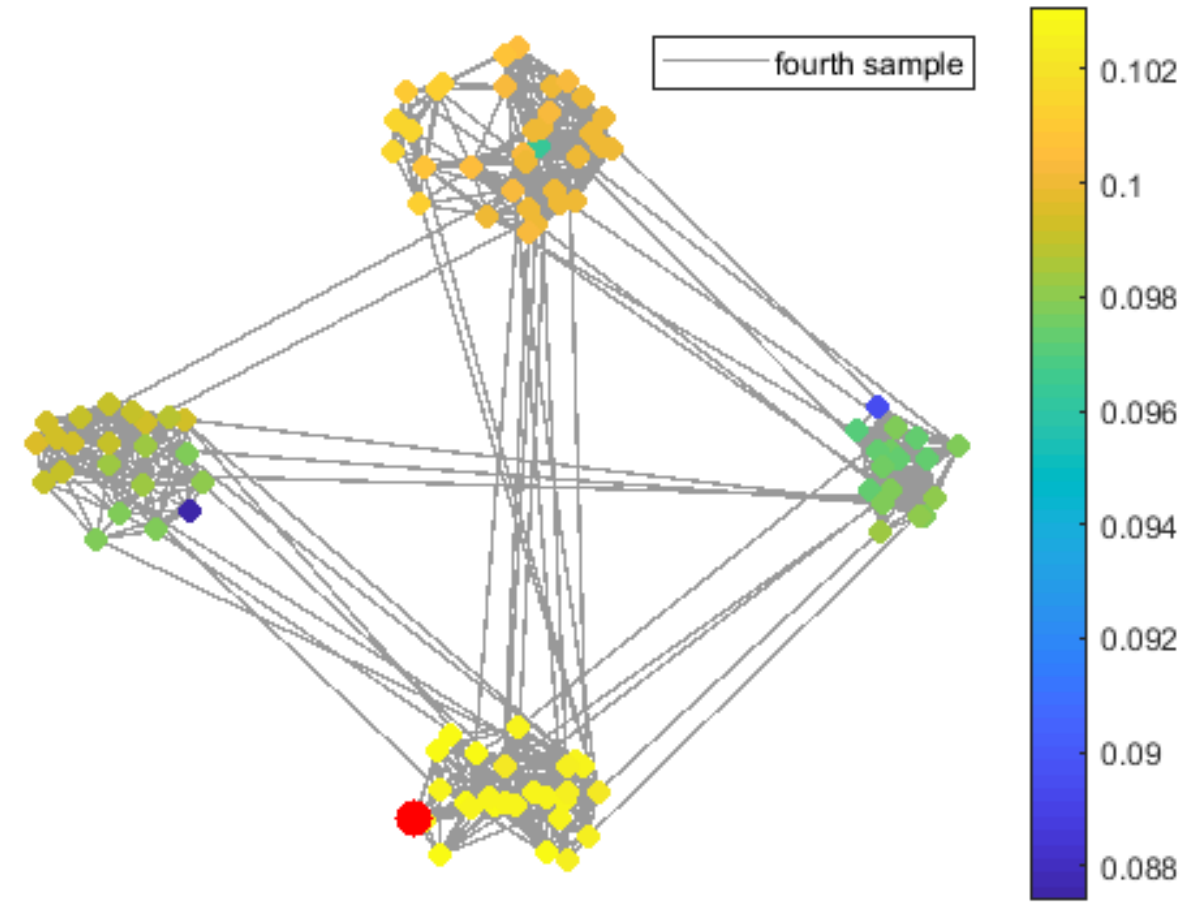}
  \caption{Sampling procedure of the proposed GCS method on a community graph with 100 nodes and 4 communities. The red circle is the sampled node in each step, whose signal energy is the largest one. %The first node is randomly selected since the the first eigenvector of $\L$ is a constant vector.
  }\label{fig:toy_GCS}
\end{figure}
This Lemma states that sampling node $i$ will reduce the spectral energy at node $i$.
Moreover, sampling node $i$ with the largest $|\bphi(i)|$ actually maximizes the upper-bound of $\lambda_{\min}(\L_t+\e_i\e^\top_i)$.
% \red{not sure what u r getting at.}
% $\bpsi$ is the first eigenvector of $\L_s$, whose elements are unknown before sampling.
%$\L_s=\L_t+\e_{i}\e^\top_{i}$ is a Laplacian with one self-loop penalty at node $i$. Without the self-loop,
%  $\bphi$ can be seen as the smoothest graph signal on graph Laplacian $\L_t$ and its energy on node $i$ is $|\bphi(i)|$.
% \red{so what?}
% Moreover, when we compute the first eigenvector $\bpsi$, we need to minimize $\y^\top\L_{t}\y+\y(i)^2$, which especially penalizes the signal's energy at node $i$.
% As we proved in Lemma \ref{lemma:bound}, the energy at node $i$ is decreased after sampling node $i$ , \textit{i.e.}, $\bpsi(i)^2\le \bphi(i)^2$.
We know that the value of $|\bpsi(i)|$ is penalized from $|\bphi(i)|$, so selecting the node with largest $|\bphi(i)|$ will also promote a reasonably large $|\bpsi(i)|$, thus provide a large lower-bound of $\lambda_{\min}(\L_t+\e_i\e^\top_i)$.

Equation \eqref{smoothness} tells us that strongly connected nodes would have similar signal, thus the energy of nodes near $i$ is also decreased when node $i$ is sampled.
Therefore, nodes close to $i$ will not be sampled by the proposed GCS method with highly probability.
This agrees with our intuition: \textit{node $i$ carries information of its local neighborhood; after sampling it, there is no need to sample connected nodes in its neighborhood.}

We conduct toy experiments on a community graph with 100 nodes using the proposed GCS sampling, whose results are shown in  Fig.\ref{fig:toy_GCS}. As depicted in this experiment, the first four samples lie in four different communities.
From the graph energy perspective, sampling one node in one community will decrease the energy of nodes within this community, {thus leading to sampling the next node  from other communities.}

% thus raising the probability of sampling the next node from other communities.
% \red{I am not comfortable using the word ``probability". our sampling strategy is not ranom / probabilistic.} 

%% file: algorithm1.tex
\begin{algorithm}[tp]
\caption{Proposed GCS Sampling Algorithm}
\label{GCS}
\textbf{Input}: Sample budget $K$; $\L=\alpha\mathbf{I}_n\otimes\mathbf{L}_{r}+\beta\mathbf{L}_{c}\otimes\mathbf{I}_m$; random vector $\mathbf{v}$\\
\textbf{Initialization}: $\cS=\emptyset$
\begin{algorithmic}[1]
\State \textbf{While} $|\cS|<K$
\State compute the first eigenvector $\pmb{\phi}$ of $\L$ with initial guess $\mathbf{v}$
\State $i^* \gets \max_{i\in\cS^c}|\bphi(i)|$
\State $\cS \gets \cS\cup\{i^*\}$
\State Update $\L=\L+\e_{i^*}\e^\top_{i^*}$ and $\mathbf{v}=\bphi$
\State \textbf{end While}
\State return $\cS$
\end{algorithmic}
\end{algorithm}

%% file: sampling2.tex
Though we identify samples using LOBPCG to compute first eigenvectors repeatedly with warm start, sampling on a very large product graph (matrix $\L$) with $mn$ nodes %\red{check} 
is still expensive for large real-world MC datasets.
We thus propose an efficient block-wise sampling method for MC  problem operating on two corresponding row and column graphs, while retaining the same sampling idea in GCS. 
%\red{check}
%\subsection{Matrix Decomposition for Fast Sampling}
Towards a simpler presentation, we omit $\Omega$ in $\tilde{\A}_{\Omega}$ in the sequel. 

Since coefficient matrix $\Q$ is a combination of row and column from $\L_r$ and $\L_c$, it does not exhibit any structure that one can exploit for optimization. 
We thus split $\Q$ into two separate matrices $\Q_1$ and $\Q_2$ as follows:
\begin{align}
&\mathbf{Q}=\left(q\tilde{\mathbf{A}}+\alpha\mathbf{I}_n\otimes\mathbf{L}_{r}\right)
+\left((1-q)\tilde{\mathbf{A}}
+\beta\mathbf{L}_{c}\otimes\mathbf{I}_m\right)\nonumber\\
&
\triangleq \Q_1+\Q_2,
\end{align}
where $0<q<1$ is a split parameter. 

Since $\Q_1$ and $\Q_2$ are both Hermitian, by Weyl's inequality \cite{neumannseries}
\begin{equation}\label{Weyl}
\lambda_{\min}(\Q) \geq \lambda_{\min}(\Q_1)+\lambda_{\min}(\Q_2),
\end{equation}
which indicates that each selected sample affects respective $\lambda_{\min}$'s of $\Q_1$ and $\Q_2$ and the lower bound of $\lambda_{\min}(\Q)$. 
%\red{check} 

From a Gershgorin circle perspective, each selected sample shifts one disc in $\alpha\I_n\otimes \L_r$ and $\beta\L_c\otimes\I_m$ by $q$ and $1-q$, respectively. 
Next, we will exploit the \textit{block diagonal} property of matrices $\Q_1$ and $\Q_2$ to develop an efficient sampling framework.  

\subsection{Block Diagonal Structure and Inner Connections}

$\Q_1\in \mathbb{R}^{mn\times mn}$ has block-diagonal structure, \textit{i.e.},
\begin{align}\label{Q1}
&\Q_1 = q\tilde{\mathbf{A}}+\alpha\mathbf{I}_n\otimes\mathbf{L}_{r}\\
&=q\left[ \begin{array}{cccc}
 \tilde{\mathbf{A}}_1& 0 & \cdots & 0 \\
0&\tilde{\mathbf{A}}_2&  & \vdots\\
\vdots & & \ddots &0\\
0 & \cdots & \cdots&\tilde{\mathbf{A}}_n
 \end{array} \right]+\alpha
\left[ \begin{array}{cccc}
\L_r & 0 & \cdots & 0 \\
0 & \L_r &  & \vdots \\
\vdots & & \ddots &0\\
0 & \cdots & \cdots&\L_r
\end{array} \right]
\nonumber
\end{align}
where $\tilde{\A}_j\in\mathbb{R}^{m\times m}$ is the $j$-th diagonal block of $\tilde{\A}$.

When entry $(i,j)$ of the target matrix $\X$ is sampled, $\tilde{\A}_j(i,i)=1$ since $\mathbf{A}(i,j)=1$ and  $\tilde{\mathbf{A}}=\text{diag}(\text{vec}({\mathbf{A}}))$. 
Equivalently, if $\tilde{\A}_j(i,i)=1$, we know that the information from $i$-th row (movie) and $j$-th column (customer) is collected. 

In contrast, matrix  $\Q_2 \in \mathbb{R}^{mn \times mn}$ is:
\begin{align}\label{Q2}
&
\Q_2=(1-q)\tilde{\mathbf{A}}+\beta\mathbf{L}_{c}\otimes\mathbf{I}_m=(1-q)\tilde{\mathbf{A}}\\
&+\beta
\left[ \begin{array}{cccc}
 \mathbf{L}_c(1,1)\mathbf{I}_m & \mathbf{L}_c(1,2)\mathbf{I}_m & \cdots & \mathbf{L}_c(1,n) \mathbf{I}_m \\
 \mathbf{L}_c(2,1)\mathbf{I}_m & \mathbf{L}_c(2,2)\mathbf{I}_m & \cdots & \mathbf{L}_c(2,n) \mathbf{I}_m \\
 \vdots & \vdots & \ddots & \vdots \\
\mathbf{L}_c(n,1)\mathbf{I}_m & \mathbf{L}_c(n,2)\mathbf{I}_m & \cdots & \mathbf{L}_c(n,n)\mathbf{I}_m
 \end{array} \right]\nonumber
\end{align}

It is known that matrix  $\L_{c}\otimes\I_m$ and $\I_m\otimes\L_{c}$ are permutation similar, \textit{i.e.}, there exists a  permutation matrix $\P$ such that \cite{Q2perm}: 
\begin{equation}
    \P(\L_{c}\otimes\I_m)\P^\top=\I_m\otimes\L_{c}
\end{equation}

\begin{figure}
  \centering
  \includegraphics[width=190pt,height=120pt]{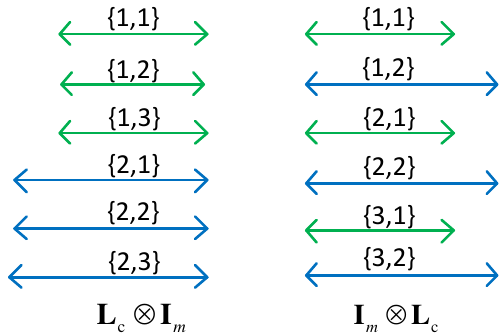}
  \caption{Gershgorin discs of matrix $\L_{c}\otimes\I_m$ (left one) and matrix $\I_m\otimes\L_{c}$ (right one), where $\L_c\in \mathbb{R}^{2\times 2}$ and $m=3$.}\label{fig:Q2perm}
\end{figure}

Thus the permuted sampling matrix $\hat{\A}$ for $\I_m\otimes\L_{c}$ is also a block-diagonal matrix. 
%\red{check}
 \begin{equation}
   \hat{\A}=\P\tilde{\A}\P^\top
   =\left[ \begin{array}{l}
 \hat{\mathbf{A}}_1 \\
 ~~~~~\hat{\mathbf{A}}_2\\
 ~~~~~~~~~~~~\ddots\\
 ~~~~~~~~~~~~~~~~~~\hat{\mathbf{A}}_m
 \end{array} \right]
 \end{equation}
where $\hat{\A}_i\in \mathbb{R}^{n\times n}$ is the $i$-th diagonal block of $\hat{\A}$.

Combined with the block diagonal property of $\I_m\otimes\L_{c}$, we will write the permuted form of $\Q_2$ as follows: 
 \begin{align}\label{Q2}
&\hat{\Q}_2=\P\Q_2\P^\top= (1-q)\hat{\mathbf{A}}+\beta\mathbf{I}_m\otimes\mathbf{L}_{c}\\
&=\hat{q}\left[ \begin{array}{cccc}
 \hat{\mathbf{A}}_1& 0 & \cdots & 0 \\
0&\hat{\mathbf{A}}_2&  & \vdots\\
\vdots & & \ddots &0\\
0 & \cdots & \cdots&\hat{\mathbf{A}}_m
 \end{array} \right]+\beta
\left[ \begin{array}{cccc}
\L_c & 0 & \cdots & 0 \\
0 & \L_c &  & \vdots \\
\vdots & & \ddots &0\\
0 & \cdots & \cdots&\L_c
\end{array} \right]
\nonumber
\end{align}
where $\hat{q}=1-q$ for brevity. 

From the property of similarity transform, we know that $\lambda(\Q_2)=\lambda(\hat{\Q}_2)$ since $\P^{-1}=\P^\top$ for any permutation matrices. 
From \eqref{Q2}, we see that when $\hat{\A}_i(j,j)=1$, the $j$-th disc in $i$-th block in matrix $\I_m\otimes\L_{c}$ is shifted. 
%However, \textit{what is the index of the corresponding shifted disc in matrix $\L_{c}\otimes\I_m$}? 
%\red{this seems more verbose than necessary.}

{We illustrate the relationship between matrix $\hat{\A}$ and $\tilde{\A}$ via a simple example.  }
Assuming $m=3,\L_c\in\mathbb{R}^{2\times2}$, we have the following two matrices: 
\begin{align}
   \L_{c}\otimes\I_3
   =\left[ \begin{array}{cccccc}
 l_{11}&~&~&l_{12}&~&~\\
 ~&l_{11}&~&~&l_{12}&~\\
 ~&~&l_{11}&~&~&l_{12}\\
  l_{21}&~&~&l_{22}&~&~\\
 ~&l_{21}&~&~&l_{22}&~\\
 ~&~&l_{21}&~&~&l_{22}
 \end{array} \right]
\end{align}
\begin{align}
   \I_3\otimes\L_{c}
   =\left[ \begin{array}{cccccc}
 l_{11}&l_{12}&~&~&~&~ \\
 l_{21}&l_{22}&~&~&~&~\\
~&~&  l_{11}&l_{12}&~&~ \\
~&~&  l_{21}&l_{22}&~&~\\
~&~&~&~&  l_{11}&l_{12}\\
~&~&~&~& l_{21}&l_{22}
 \end{array} \right]
\end{align}
where $l_{ij}$ are the elements in matrix $\L_c$. 

 We first assign the Gershgorin discs in matrices $\L_{c}\otimes\I_3$ and  $\I_3\otimes\L_{c}$ with indices, 
% indices $\{\{1,1\};\{1,2\};\{1,3\};\{2,1\};\{2,2\};\{2,3\}\}$ and  $\{\{1,1\};\{1,2\};\{2,1\};\{2,2\};\{3,1\};\{3,2\}\}$ respectively, 
as illustrated in Fig.\;\ref{fig:Q2perm}. Assuming $\tilde{\A}_1(3,3)=1$, the disc $\{1,3\}$ in matrix $\L_{c}\otimes\I_3$ is shifted by $\hat{q}$. After permutation, this means that the disc $\{3,1\}$ in matrix $\I_3\otimes\L_{c}$ is shifted, \textit{i.e.}, $\hat{\A}_3(1,1)=1$.  
Thus, $\tilde{\A}_j(i,i)=1$ is equivalent to  $\hat{\A}_i(j,j)=1$. 

Therefore, {sampling data $\X$ at $(i,j)$ will promote the smallest eigenvalue of $j$-th ($i$-th) diagonal block of $\Q_1$ {($\hat{\Q}_2$)} by making $\tilde{\A}_j(i,i)=1$ {($\hat{\A}_i(j,j)=1$)}.}
Given this connection, we next propose an iterative sampling strategy. 
Block matrices in $\Q_1$ and $\hat{\Q}_2$ will be called \emph{`clusters'} and   \emph{`groups'} respectively. 

\subsection{Iterative Sampling between Clusters and Groups}
\label{subsec:sampling2}
% Sampling one signal from a matrix data (querying one customer's rating of one movie) will benefit the eigenvalues of one cluster in $\Q_1$ and one group in $\hat{\Q}_2$ in the meantime, with shift $q$ and $1-q$ respectively.  
% Our original sampling objective is to raise the smallest eigenvalue of $\Q$ as much as possible. 
% Promoting the smallest eigenvalue of each cluster and group as much as possible greedily is meaningful for 
We here propose to alternately  collect  samples based on one cluster in $\Q_1$ only or one group in $\hat{\Q}_2$ only. 
Specifically, we start sampling the matrix signal from the first column ($j=1$) based on the Laplacian matrix of the first cluster in $\Q_1$. 
%\red{not clear to me.}
If its  first eigenvetor has the largest energy at the $i$-th index, we will sample data at $(i,1)$ and then proceed sampling based on corresponding incremented Laplacian matrix of the $i$-th group in $\hat{\Q}_2$. 
We continue to choose samples alternating between clusters and groups until the sampling budget is exhausted. 
%This iterative sampling framework is illustrated in Algorithm \ref{IGCS}, which is called \emph{iterative Gershgorin circle shift} (IGCS) sampling method. 
% \begin{figure}
%   \centering
%   \includegraphics[width=140pt,height=130pt]{iterativeSampling.pdf}
%   \caption{Sampling location reflection of sampling  based on information of one cluster and then of one group iteratively.}\label{iterative}
% \end{figure}

Since our GCS sampling strategy using LOBPCG benefits from warm start, the computation complexity of this iterative scheme can be further reduced if we choose more than one sample from the same cluster (or group). 
We thus introduce a warm start parameter $\zeta$ to trade off sampling performance and computation complexity; its sensitivity will be examined in Section \ref{sec:results}.
%\red{$k$ is a poor choice of parameter name. it has been used before and typically it is an interation / matrix / vector index, not a parameter.}
Detailed iterative sampling pseudo-code is shown in Algorithm \ref{IGCS}, called \emph{iterative Gershgorin circle shift} (IGCS)-based sampling. 
As we analyzed in Section\;\ref{sec:sampling1}, our proposed GCS has complexity $\cO(KFmn)$, while the complexity of IGCS is just $\cO(K\hat{F}c)$.  $\hat{F}$ is the convergence iteration number of LOBPCG in IGCS, and $c=\max\{m,n\}$. 
{Therefore, the complexity is reduced by at least a factor $\min\{m,n\}$ using this iterative sampling framework, \textit{i.e.}, the complexity is roughly linear to the size of factor graph.}

\begin{algorithm}[tp]
\caption{Proposed IGCS Sampling Algorithm} 
\label{IGCS}
\textbf{Input}:$K$, $\L_r$, $\L_c$, $q$, $\alpha$, $\beta$ and warm start parameter $\zeta$\\
\textbf{Initialization}: $\Omega=\emptyset$,  
$\tilde{\A}_j=\mathbf{0}$,  $\hat{\A}_i=\mathbf{0}$, $j=1$, $s=1$ and $w=0$,  $\hat{q}=1-q$
\begin{algorithmic}[1]
\State \textbf{while} $|\Omega|<K$
\If {$s=1$} \Comment{sample from $j$-th cluster}
\State $\tilde{\L}=q\tilde{\A}_j+\alpha\L_r$;
$\cS=\{t|\tilde{\A}_j(t,t)=1\}$; $w=w+1$
\State If $w=1$, random $\v\in \mathbb{R}^m$; else, $\v=\bphi$
\State compute the first eigenvector $\bphi$ of $\tilde{\L}$ with input $\v$ 
\State $k^* \gets \max_{k\in\cV_r\setminus\cS}|\bphi(k)|$
\State $\Omega\gets \Omega\cup\{(k^*,j)\}$
\State $\tilde{\A}_j(k^*,k^*)=1$; $\hat{\A}_{k^*}(j,j)=1$
\State If $w\ge \zeta$, then $i=k^*$, $s=2$ and $w=0$
\Else  \Comment{sample from $i$-th group}
\State $\hat{\L}=\hat{q}\hat{\A}_i+\beta\L_c$;  $\cS=\{t|\hat{\A}_i(t,t)=1\}$; $w=w+1$
\State If $w=1$, random $\v\in \mathbb{R}^n$; else, $\v=\bphi$
\State compute the first eigenvector $\pmb{\phi}$ of $\hat{\L}$ with input $\v$
\State $k^* \gets \max_{k\in\cV_c\setminus\cS}|\bphi(k)|$
\State $\Omega\gets \Omega\cup\{(i,k^*)\}$
\State $\tilde{\A}_{k^*}(i,i)=1$; $\hat{\A}_i(k^*,k^*)=1$
\State If $w\ge \zeta$, then $j=k^*$, $s=1$ and $w=0$
\EndIf
\State \textbf{end while}
\State return $\Omega$
\end{algorithmic}
\end{algorithm}

% From a sampling view, assume the first sampled data is in $(i,j)$, it means $j$-th customer has provided his rating on the $i$-th movie. 
% In the next sampling step, we would like to know which movie is the next most informative one for this customer, then we will query the rating from him for that movie. We don't want to immediately query a totally new people since switching customers is expensive and time-consuming.
% \red{not sure this is a strong argument.}
% Then we will go to the movie domain, for the already rated movie, who is most valuable people to ask for the rating? 
% The iterative querying strategy meets our prior knowledge about how the query system works to some extent. 
% The iterative sampling procedure in the data domain is shown in Fig.\ref{iterative}. 
% Assume the IGCS method selects the first node at location ``1'' based on the first cluster's information, and then it will keep sampling on this row (``1''$\rightarrow$``2''), i.e., sampling based on one group's information, and then switch sampling on the column (``2''$\rightarrow$``3''), i.e., sampling based on clusters. 

%% file: graph.tex
Before using graph sampling for MC via the proposed IGCS, one has to first acquire the row graph $\L_r$ and the column graph $\L_c$. 
There exist many methods to construct finite row / column graphs from data, so that the observed signal(s) are smooth (low-pass) with respect to the constructed graphs \cite{EPFLgraph,huang2018rating}. 
For completeness, we overview methods we chose to construct row and column graphs using which we select samples. 
We stress that our work focuses on sampling; the discussion here merely demonstrates that our graph sampling schemes can be practically realized in combination with existing graph learning methods. 
%We use the following published heuristic methods to generate row and column graphs: 

$\bullet$ G1: \textbf{Feature-based graph}

As done in graph-based MC methods \cite{monti2017geometric,berg2017graph}, when the user / item profiles (\textit{e.g.}, age, gender and occupation of users and genre of the items) are available, we construct a weighted 10-nearest neighbor graph using GSPBox \cite{gspbox} based on feature vectors of each node. 

$\bullet$ G2: \textbf{Content-based graph from observed information}

When features of data points are not available, we construct row and column graphs only from partial matrix entries, extending method used in \cite{EPFLgraph}. 
Specifically, the observed matrix is $\Z=\A_{\Gamma}\circ\X$ for a given random initial set $\Gamma$. Then, for  each pair of users $\{i,j\}$, their partial ratings are in the $i$- and $j$-th rows of matrix $\Z$, denoted by $\z_i$ and $\z_j$. 
We then compute the inter-node distance as  \begin{align}
    d_{ij}=\frac{||\mathbf{z}_i(\mathcal{R}_{ij})-\mathbf{z}_j(\mathcal{R}_{ij})||_2}{\sqrt{|\cR_{ij}|}}
\end{align} 
where $\mathcal{R}_{ij}=\mathcal{R}_i\cap\mathcal{R}_j$, and $\cR_i$ is the set of items rated by user $i$. 

If $|\cR_{ij}|=0$, we set $d_{ij}=\infty$.
We then compute the edge weight between users $i$ and $j$ as 
\begin{equation}\label{sampling}
\begin{split}
w_{ij} = \left\{ \begin{array}{ll}
\exp\{-(d_{ij}-d_{\min})^2/\gamma\};&\mbox{if}\; d_{ij}\le d_s \\
0,&\mbox{otherwise}
\end{array} \right.
\end{split}
\end{equation}
where $d_{\min}=\min_{\{i,j\}}d_{ij}$ and $d_s$ is the threshold of user $i$ for sparsifying $\W_r$;  $\gamma$ is a factor to control function shape for weight computation. 

Likewise, the item graph is constructed similarly using column ratings in observed matrix $\Z$.
In our experiments, for real-world datasets, we first assume partial ground-truth data is known, and then use them to construct the factor graphs based on the above method. 
With the constructed graphs, we proceed the following sampling based on different schemes and then compute the completion error on the unobserved entries. 

%2) We use the densest $100\times 200$ submatrix from large dataset as the  experimental input in Fig. \ref{fig:small_dataset} (c) and (d).  By assuming that the information outside this submatrix is given as prior, we construct the graph via the above-mentioned method, which is the same as paper \cite{EPFLgraph}.

% %In experiments, the percentage ($p$) of randomly known observations is set based on the size of the original large dataset. 

%% file: results.tex
\begin{table*}
\centering
\caption{Profiles of experimented datasets.}
\label{tab:dataset}
\begin{tabular}{lcccccccc}
\toprule
dataset  &Exp.      & users & items & features & entries & density  & entry levels \\\midrule
\texttt{Synthetic Netflix} \cite{EPFLgraph}&Fig.\;\ref{fig:small_data}  & 200     & 100      &    -      &  20,000       &     100\%    &   1,2,...,5      \\ \midrule
\multirow{2}{*}[-0.5ex]{\texttt{ML100K} \cite{movielens}}&Fig.\;\ref{fig:small_data} &   100 &200       &    -      &  12,566      &   62.83\%      &   1,2,...,5           \\
 &Tab.\;\ref{tab:movielens_100k1}/\ref{tab:movielens_100k2}& 943   &    1682   &    \checkmark      &  100,000       & 6.3\%        &   1,2,...,5           \\ \midrule
\multirow{2}{*}[-0.5ex]{\texttt{ML10M} \cite{movielens}} &
Fig.\;\ref{fig:small_data}&100   &200       &       -   &  18,119       &   90.6\%      &   0.5,1,...,5           \\
 & Tab.\;\ref{tab:datasets}& 1000  & 500      &       -   &   345,904      &     69.18\%    &   0.5,1,...,5           \\ \midrule
\texttt{Douban} \cite{berg2017graph}&Tab.\;\ref{tab:datasets}   & 3000   &    3000   &   - & 136,891       & 1.52\%   &   1,2,...,5           \\
\texttt{Flixster} \cite{berg2017graph}&Tab.\;\ref{tab:datasets}        &  3000  & 3000      &     - & 26,173        &     0.29\%    &  0.5,1,...,5            \\
\texttt{YahooMusic} \cite{berg2017graph} &Tab.\;\ref{tab:datasets}     & 3000  & 3000      &  -  & 5,335   &  0.06\%       &   1,2,...,100           \\
\texttt{Book-Crossing} \cite{book_crossing}&Tab.\;\ref{tab:datasets}  &1000   & 1000      &    -      &  3,166       &     0.32\%    &   1,2,...,10          \\
\texttt{Jester} \cite{Jester}&Tab.\;\ref{tab:datasets} &1000 &100
&- &73,320 & 73.32\%& (0,1)
\\
\texttt{ML1M} \cite{movielens}   & Tab.\;\ref{tab:datasets}& 6040  & 3706      &  \checkmark         &  1,000,209       &   4.47\%      &  1,2,...,5            \\
\texttt{FilmTrust} \cite{Filmtrust} & Tab.\;\ref{tab:datasets} & 1000&1000&-&31,880&3.19\%& 0.5,1,...,4\\
\bottomrule
\end{tabular}
\end{table*}

\begin{figure*}
    \centering
            \subfigure[Nosiy synthetic Netflix signal]{
    \includegraphics[width=150pt,height=150pt]{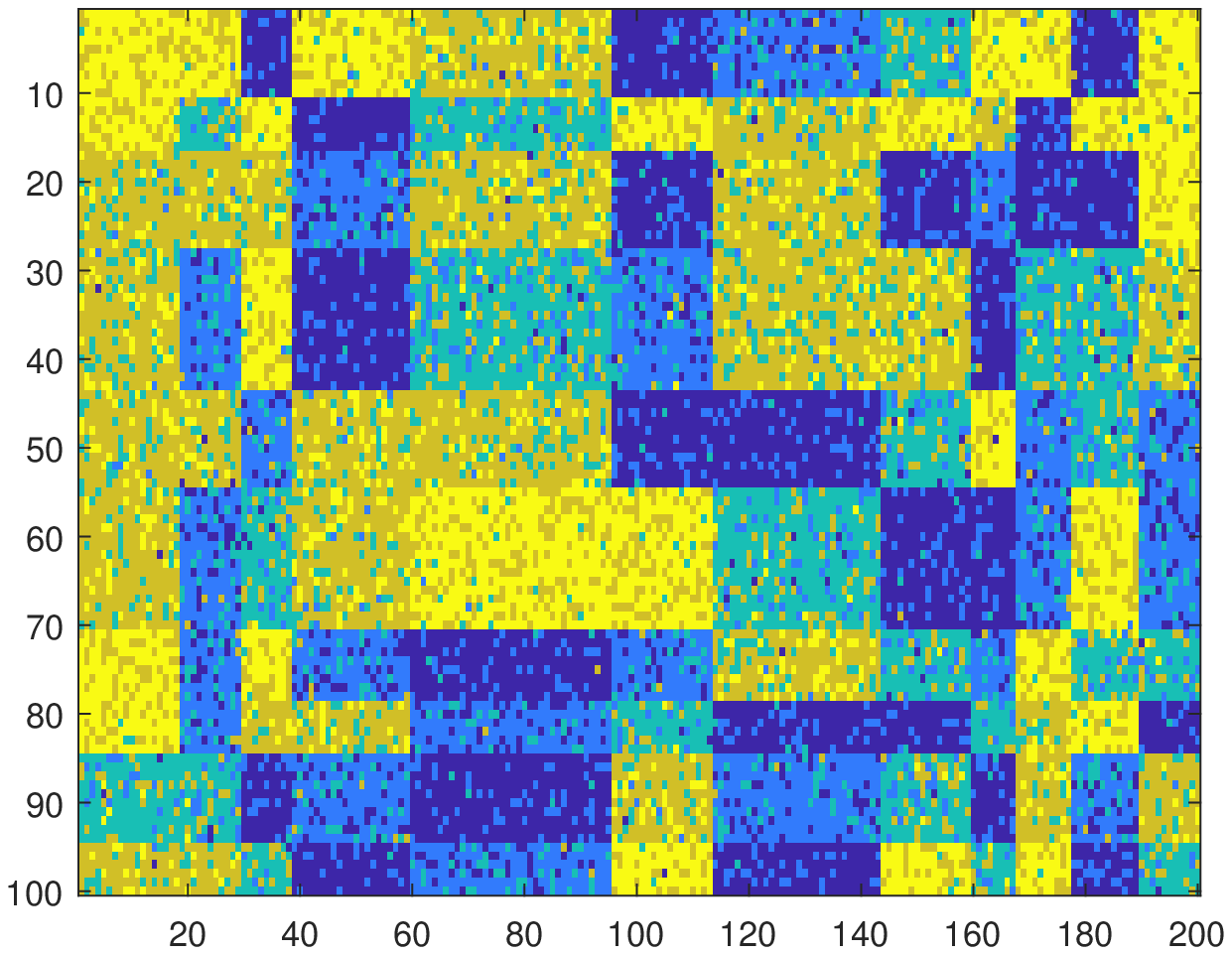}}
        \subfigure[RMSE on noiseless signal]{
    \includegraphics[width=150pt,height=150pt]{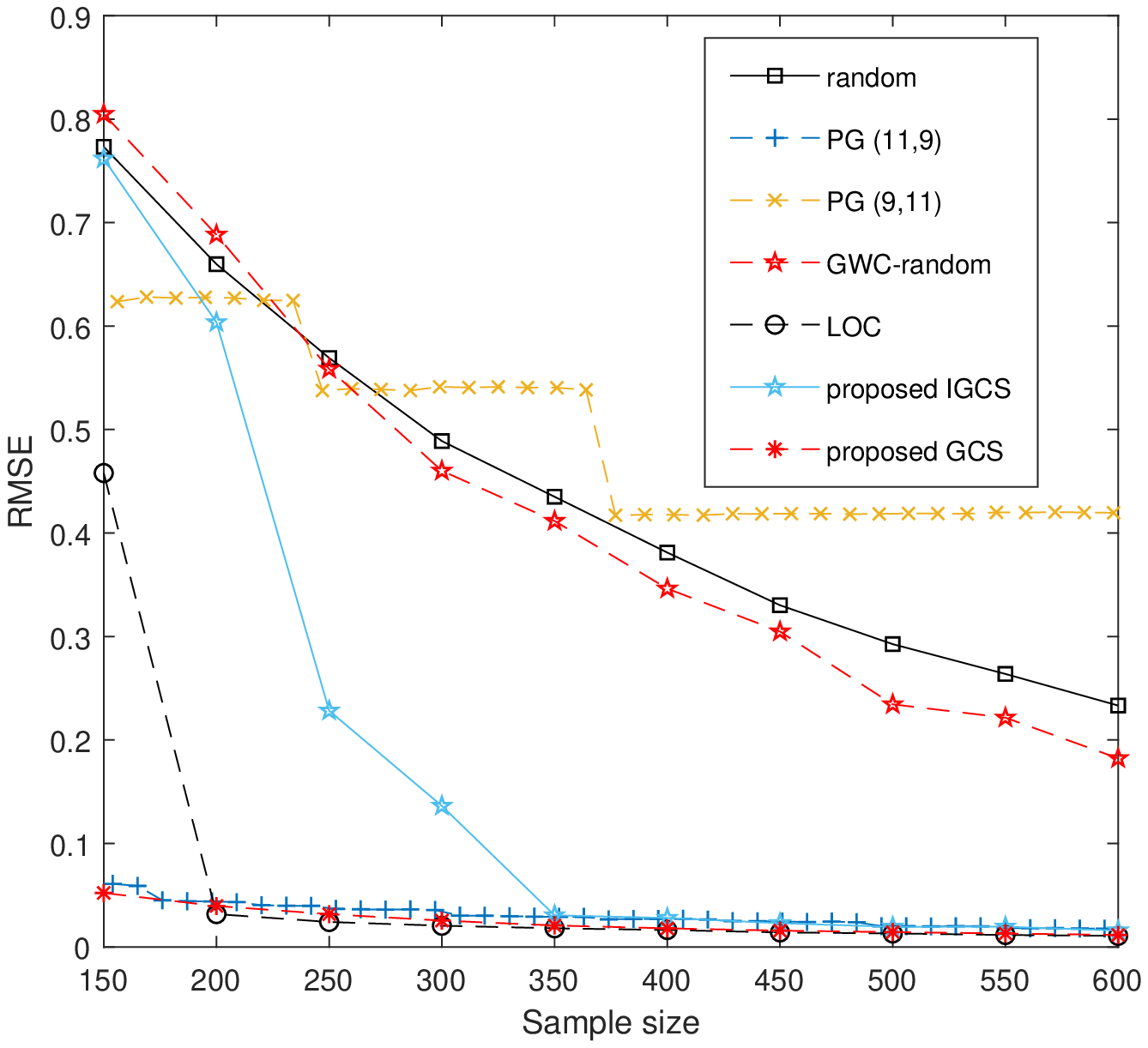}}
        \subfigure[RMSE on noisy signal with $\gamma=0.6$]{
    \includegraphics[width=150pt,height=150pt]{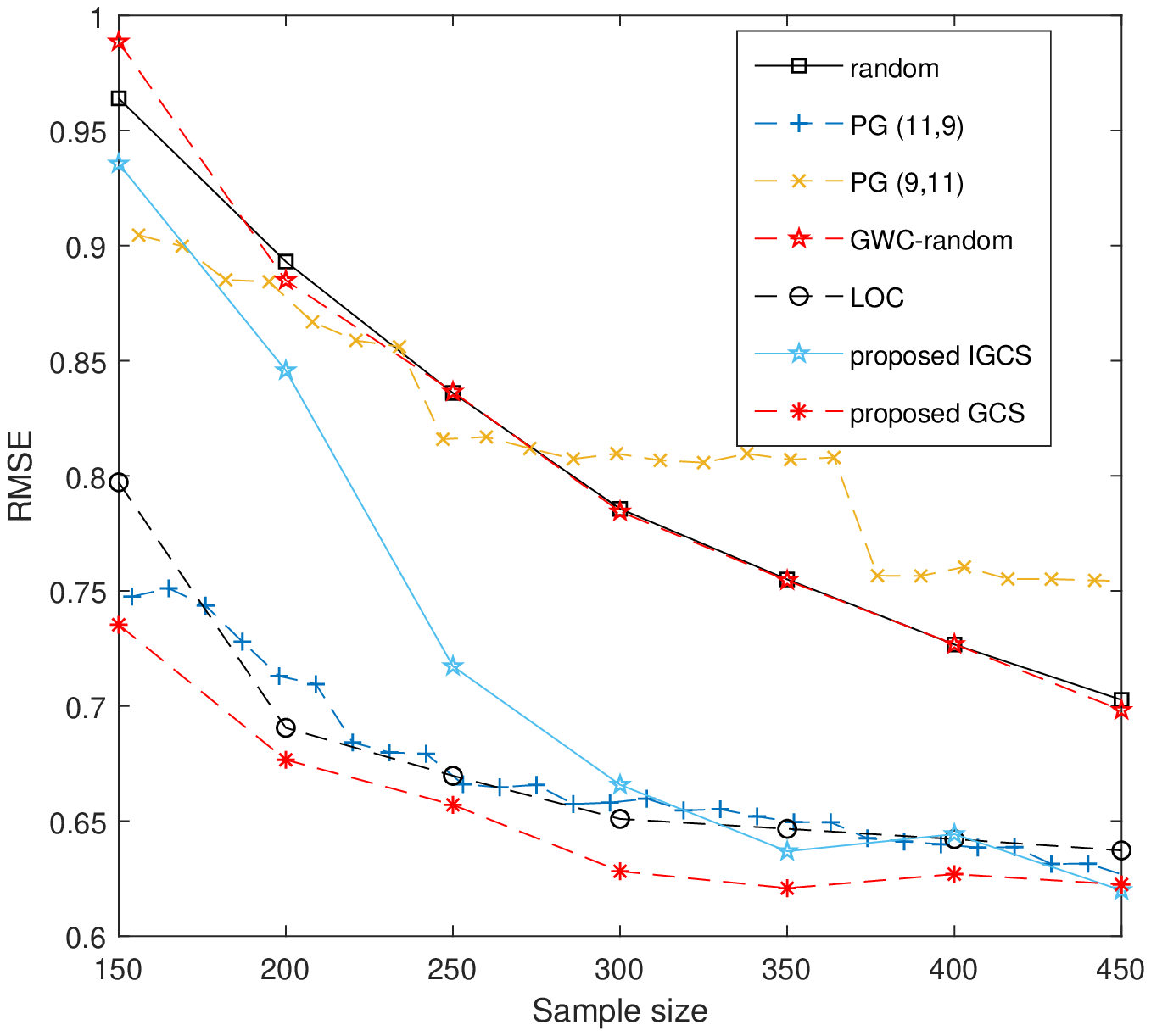}}
    \subfigure[RMSE on different noise level $\gamma$ with $K=150$]{
    \includegraphics[width=150pt,height=150pt]{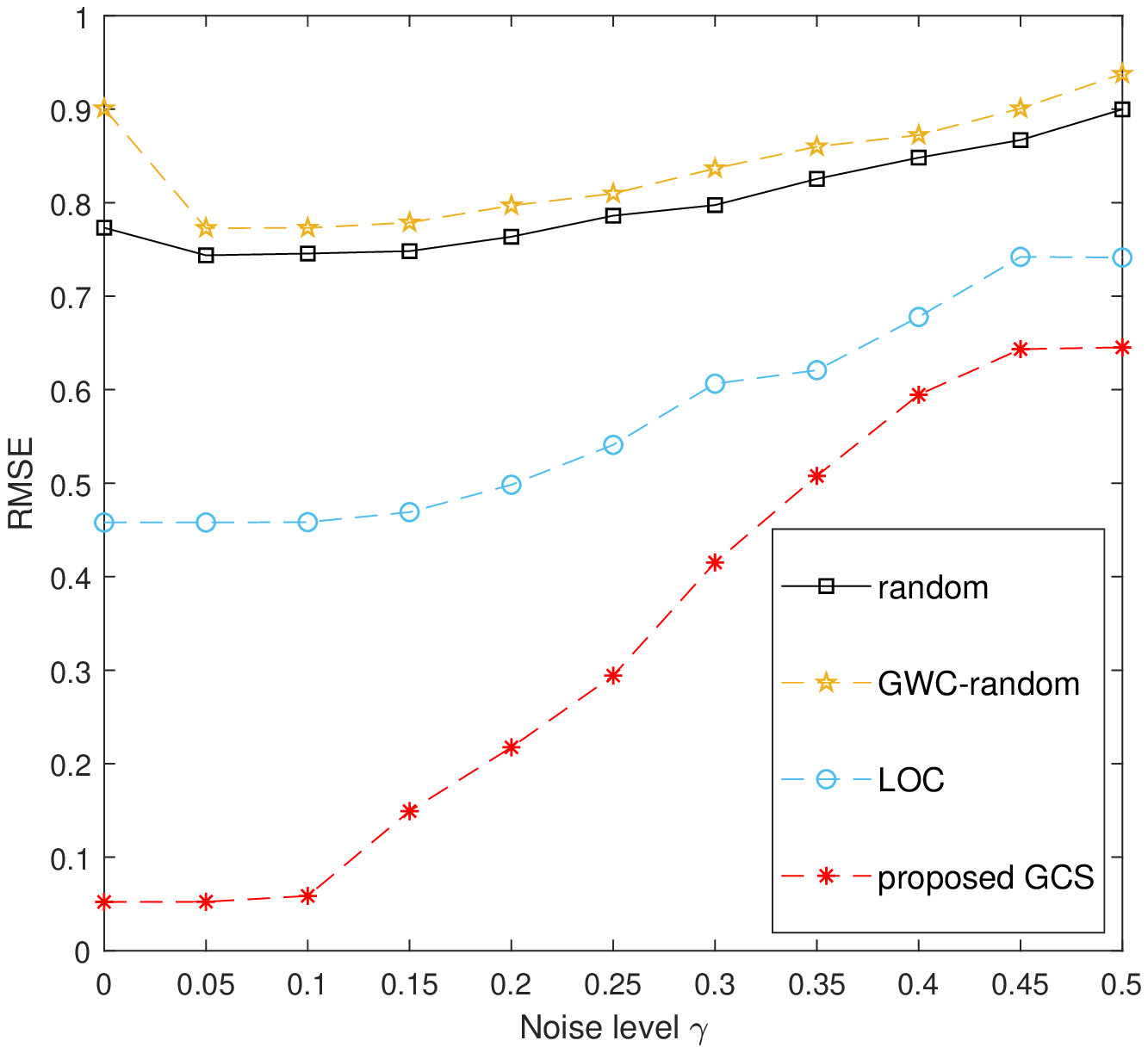}}
    \subfigure[\texttt{ML100K} ($100\times 200$)]{
    \includegraphics[width=150pt,height=150pt]{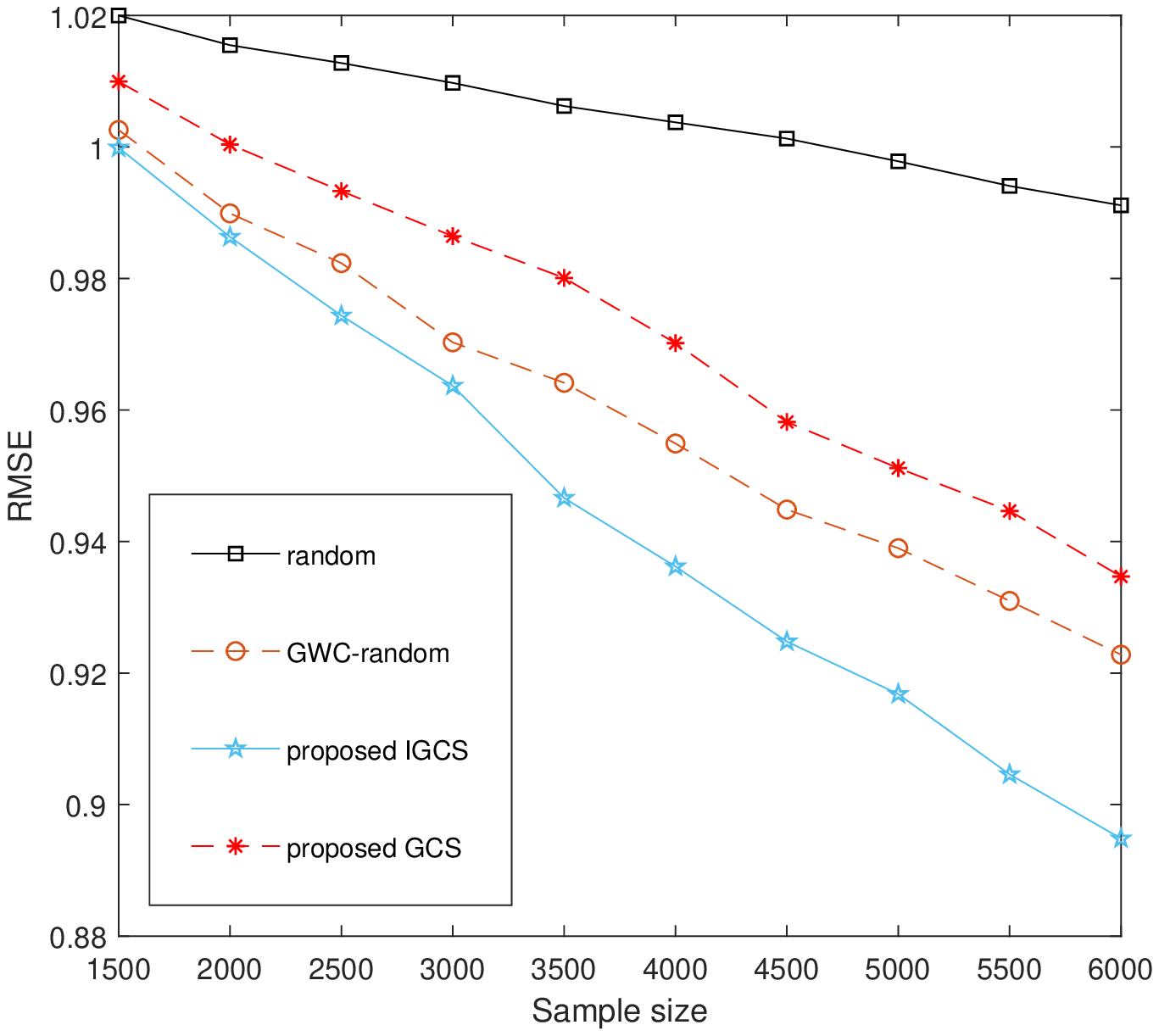}}
        \subfigure[\texttt{ML10M} ($100\times 200$)]{
    \includegraphics[width=150pt,height=150pt]{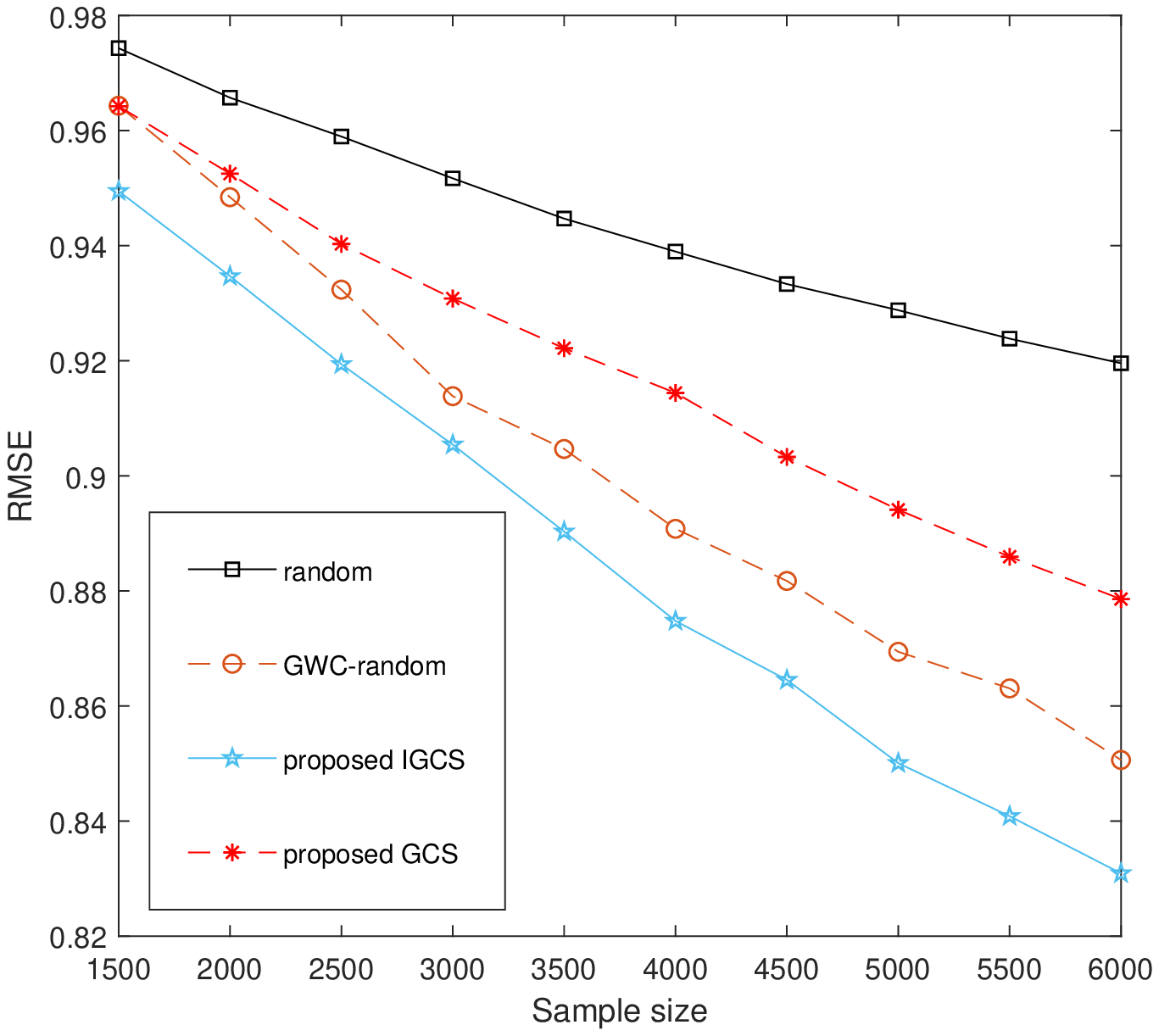}}
        \caption{RMSE of different sampling methods for MC on \texttt{Synthetic Netflix (SN)} \cite{EPFLgraph}, \texttt{ML100K} ($100\times 200$) and \texttt{ML10M} ($100\times 200$) datasets. The matrix is completed by dual smoothness based method \eqref{eq:obj0}.}
        \label{fig:small_data}
\end{figure*}

In this section, we present experimental results of our proposed sampling methods and other competing schemes, combined with several state-of-the-art MC strategies.
% All sampling experiments were performed in MATLAB R2018a, running on a PC with Intel Xeon Silver 4110 2.1 GHz CPU and 32GB RAM.
We list profiles of the simulated datasets in Table.\;\ref{tab:dataset}.
% Note that the ground-truth matrices signal of real-world datasets are already sparse. Thus, for those datasets, we have to implement unstructured competitors for comparison. \red{I don't get what u r talking about.}
%The experimental results are specifically described as follows.

\subsection{Experimental Setup}
In all experiments, we set $\alpha=\beta=0.1$ for GCS and IGCS, and set $q=0.5$ for IGCS.
We implement five sampling methods for comparison, whose specific settings are as follows:
\begin{itemize}
    \item \textit{Graph weight coherence} (GWC-random) \cite{puy18}:  the `estimated' setting was used for computing the probability for random sampling without replacement. The bandwidth information was set to be 1000.
    \item \textit{Localized operator coverage} (LOC) \cite{akie2018eigenFREE}: the bandwidth prior was set to be 1000.
    \item \textit{Product Graph}-based sampling (PG)  \cite{PG}: the dual graph bandwidth  $(\eta_1,\eta_2)$ was set to be $(11,9)$ or $(9,11)$ for conducting structured sampling.
    \item LSS \cite{LLS}: we set rank to be 5 for implementing this method.
\end{itemize}

The last competing method is uniform random sampling. In the following, we list the details for all simulated MC methods:
\begin{itemize}
        \item IMC \cite{IMC_code} and SVT \cite{cai2010singular} were simulated with the same settings as reproducible codes.
    \item GMC \cite{EPFLgraph}: we set $\gamma_n=3$, $\gamma_r=\alpha=0.1$ and $\gamma_c=\beta=0.1$. For completion method in equation \eqref{eq:obj0}, we set $\gamma_n=0$ and keep other parameters the same.
    \item GRALS \cite{GRALS}: we set the rank to be 5  for GRALS, and its Laplacian matrix input were computed from $\L_h=\L_c+0.1\I_n$ and $\L_w=\L_r+0.1\I_m$, as used in \cite{PG}.
    \item GCMC \cite{berg2017graph}: the training epochs were set to be 1000.
    \item NMC \cite{NMC}: the number of training epochs was 10000.
\end{itemize}
\subsection{Performance on Small Datasets}
We first conduct experiments on small-size \texttt{Synthetic Netflix} datasets for performance comparison, where the matrix is completed by dual smoothness based method \eqref{eq:obj0}.
For our proposed GCS, LOBPCG is employed with warm start to compute the first eigenvector of matrix $\L\in\mathbb{R}^{20000}$, while the IGCS uses the MATLAB's inbuilt function (Krylov-Schur method) for eigen-decomposition since the factor graphs are small.
Two classical graph sampling methods GWC-random \cite{puy18} and LOC \cite{akie2018eigenFREE} are implemented on the product graph    $\L_p=\mathbf{I}_n\otimes\mathbf{L}_{r}+\mathbf{L}_{c}\otimes\mathbf{I}_m$ directly.
For structured method PG, with bandwidth input $(11,9)$ or $(9,11)$, we artificially increase the parameter $L=|\cL_1|+|\cL_2|$ to get rectangular output, where $\cL_1$ and $\cL_2$ are its selected row and column indices, respectively. After sampling, we record its sample size $|\cL_1|\times |\cL_2|$ and corresponding reconstruction error.
The root mean square error (RMSE) is  computed on unobserved entries in terms of ground-truth value for evaluation, as done in \cite{PG,anis14,tensor}.

Specific experimental results on noiseless and noisy \texttt{synthetic Netflix} dataset (noisy one is depicted in Fig.\;\ref{fig:small_data}\,(a)) in terms of sample size are shown in Fig.\;\ref{fig:small_data}\,(b) and (c).
We observe that our proposed GCS outperforms all competitors especially when the sampling budget is small and the matrix signal is noisy.
Though with bandwidth $(11,9)$, PG has comparable RMSE value, its performance deteriorates drastically by just changing the bandwidth to $(9,11)$.
{This means that PG is very sensitive to bandwidth settings.}
Further, PG cannot achieve arbitrary sample size due to its rigid sampling structure.

{Our proposed iterative sampling method IGCS also achieves good performance when sample size is relatively large with  complexity $\cO(K\hat{F}\max\{m,n\})$. Recall that the complexity of GCS is $\cO(KFmn)$ and LOC is $\cO(KJ)$, where $J$ is the number of non-zero entries in matrix $\L^d$.
We know that $J=\cO(d_{\max}mn)$, where $d_{\max}$ is the largest degree in product graph $\L$.
Hence IGCS has much lower complexity.}
Fig.\;\ref{fig:small_data}\,(d) further illustrates GCS's superiority for different noise levels, where we remove some inferior competitors for better visualization.
%(GCS is still the best one without deletion).

%$\bullet$ \textbf{Submatrix of real-world dataset}

We also test sampling methods on the dense submatrix ($100\times 200$) from two real-world datasets \texttt{ML100K} and \texttt{ML10M}.
As done in \cite{EPFLgraph}, by assuming that the information outside this submatrix is given as prior, we construct content-based graph G2 via strategy described in Section  \ref{sec:graph-construction}.
Since the ground truth submatrix is sparse, the sampling method must be constrained to sample on the sparse entries.
PG, being a structured sampling strategy, cannot satisfy this requirement.
%In contrast, our method is an unstructured and greedy sampling strategy, which can collect sample on limited entries.
Further, LOC requires computation of a Chebyshev polynomial graph filter  before sampling, which always results in out-of-memory error using our constructed graph.
Thus, we show only executable sampling schemes for comparison.
The resulting RMSE is shown in Fig.\;\ref{fig:small_data}\,(e) and (f), which illustrate IGCS's superiority over GWC, GCS and random sampling for those constructed small real-world datasets.

 \subsection{Performance on Real-world Large Datasets}
To actively sample entries on large real-world datasets, both GCS and GWC are not applicable, since the size of the product graph $\L\in \mathbb{R}^{mn\times mn}$ can contain millions of nodes.
For the following real-world large datasets, we only test uniform random sampling and LSS \cite{LLS}  for performance comparison to our proposed IGCS.
% Please note that since the real-world dataset is sparse at the beginning, sampling methods are implemented to select samples at constrained available entries.

$\bullet$ {\textbf{Simulations on Movielens 100K}}

We first deploy our proposed IGCS on a large real-world dataset \texttt{ML100K} \cite{movielens} to collect samples.
Since the typically used ratio between training and testing in \texttt{ML100K} is 80\% to 20\%, in our experiments, {we first randomly select
60K samples from 100K datasets as the initial available data, and then proceed to sample 20K from the rest 40K data pool based on our proposed IGCS or random sampling. The final un-selected 20K samples are used for computing RMSE. }
% \red{what do u mean by randomly generate?}
% \red{I don't follow}
For this experiment, we use MATLAB's inbuilt function for eigen-decomposition in IGCS.

We create the feature-based graph G1 from   \texttt{ML100k}'s features and content-based graph G2 from 60K initial samples using the second method in Section \ref{sec:graph-construction}.
Average RMSE on different MC methods  and graphs are listed in Table.\;\ref{tab:movielens_100k1}, where the best performance number for each MC method is marked in boldface.
In the widely used feature-based graph (G1), our proposed IGCS (right side) achieves  better performance than random sampling (left side) for almost {all} popular MC methods.
%\red{what do u mean by almost popular? not popular?}
Further, when we select entries on the content-based graph (G2), IGCS substantially outperforms random sampling.
Note that when G2 is used instead of G1, RMSE for random sampling is almost the same for every graph-based MC method.
Hence, we can conclude that the performance improvement using G2 is due to the more informative samples chosen using our proposed IGCS.
%It is also known to us that side information to construct G1 is difficult to obatine

\begin{table}
    \centering
    \caption{RMSE for \texttt{ML100K} using random / IGCS sampling combined with different MC methods.
    Graph-based MC strategies are marked with  \checkmark.}
    \label{tab:movielens_100k1}
    \begin{tabular}{lcccc}
    \toprule
   MC methods&$\cG$? &G1&G2\\
    \midrule
    IMC \cite{IMC}&-&1.590 - \textbf{1.507}&1.590 - 1.600\\
    %\midrule
    SVT \cite{cai2010singular}&-&{1.021} - 1.031 &1.021 - \textbf{0.983}\\
    %\midrule
    GRALS \cite{GRALS}&\checkmark&0.947 - {0.931}&0.945 - \textbf{0.893}\\
    %\midrule
    GMC \cite{EPFLgraph}&\checkmark&\textbf{1.036} - 1.037&1.118 - {1.054}\\
    %\midrule
    GC-MC \cite{berg2017graph}&\checkmark&0.898 - {0.891}&0.899 - \textbf{0.858}\\
    %\midrule
    NMC \cite{NMC}&-&0.892 - {0.887}&0.892 - \textbf{0.861}\\
    \bottomrule
    \end{tabular}
\end{table}

$\bullet$ {\textbf{Warm start parameter's effect on sampling time}}

%In the above-mentioned experiments, we use the MATLAB's inbuilt sparse matrix eigen-decomposition method (Krylov-Schur algorithm) to compute the first eigenvector in IGCS method without warm start ($k=1$).
In this experiment, we deploy IGCS on \texttt{ML100K} using different graphs and in combination with different MC methods.
The resulting RMSE values and sampling times are shown in Table.\;\ref{tab:movielens_100k2}, along with LSS for comparison.
``eigs'' means the eigen-decompostion in IGCS is computed using the Krylov-Schur method, while LOBPCG is used for different $\zeta$.
%The sampling size is still 20K.
Note that LSS is essentially a random sampling with specified selection probability for each entry.
Thus we didn't record its running time.
All experiments are performed on a laptop with Intel Core i7-8750H and 16GB of RAM on Windows 10 for counting time.
In Table.\;\ref{tab:movielens_100k2}, the best performance numbers for each method are marked in boldface.
Table.\;\ref{tab:movielens_100k2} shows that IGCS is always superior to LSS using different state-of-the-art MC methods under different graphs.
It also shows that with increasing $\zeta$, execution time of IGCS with LOBPCG decreases substantially, while the performance become slightly worse.
Note that when $\zeta=1$, there is no warm start in LOBPCG.
Simulation results show that LOBPCG is more efficient than Krylov-Schur method for computing the first eigenvector and achieves better performance for MC.

%Moreover, the performance and sampling time of LOBPCG-based IGCS is respectively better and less than the Krylov-Schur-based IGCS method for  the simulated matrix completion methods and graphs.
%Note that the sampling time in G1, G2 and G3 are different under the same $k$, this is because the sparsity of those three graphs are different. Specifically, in G1, $|\cE_r|=5757$ and $|\cE_c|=15408$; in G2, $|\cE_r|=5923$ and $|\cE_c|=11705$; in G3, $|\cE_r|=79783$ and $|\cE_c|=201899$. Therefore, in general, with the same $k$, G3 will cost the most sampling time while the G2 will cost the least sampling time for the LOBPCG-based IGCS, which meet the experimental results.
$\bullet$ {\textbf{Simulations on other popular real-world datasets}}

We next evaluate IGCS on various well-known real-world datasets, combined with GRALS MC method.
Random sampling and LSS are simulated for comparison.
Since features for constructing G1 are not available for most datasets, we use G2 as the underlying graph for sampling.
%and matrix completion.
For datasets \texttt{Flixter}, \texttt{YahooMusic}, \texttt{Douban}, random 90/10 training/test splits are used for simulations.
Specifically, we first choose 80\% entries in the given 90\%  training set as the initial samples to construct G2 and then use our IGCS method (or competing schemes) to sample entries to form a new 90\% training set.
RMSE is computed on final un-selected 10\% entries. %, which is different among different sampling strategies.
For other datasets, we first randomly generate 90/10 training/test split and then use the above-mentioned procedure to collect samples and compute RMSE \footnote{
{For datasets \texttt{ML1M} and \texttt{ML10M}, the percentage of initial samples for constructing G2  is changed from 80\% into 90\%. }}.

\input{table4.tex}
\begin{table}[tp]
    \caption{RMSE of the proposed IGCS on different datasets with different $\zeta$'s, along with random sampling and LSS for comparison. The MC method is GRALS.}
    \centering
    \label{tab:datasets}
    \begin{scriptsize}
\begin{tabular}{lcccccc}
\toprule
 dataset                   & random &LSS& $\zeta=1$    & $\zeta=3$    & $\zeta=5$    & $\zeta=7$    \\\midrule
\texttt{Flixster}                   & 1.029 &1.207 & \bf{0.932}  & 1.057  & 1.046  & 1.045    \\
\texttt{Douban}                     & 0.744  &0.750& \bf{0.715}  & 0.720  & 0.736   & 0.730   \\
\texttt{YahooMusic}                   & 96.987&125.0 & 59.172 & \bf{44.546} & 52.391 & 47.082 \\
\texttt{ML1M}                     & 0.905  &0.930& \bf{0.829}  & 0.833  & 0.835  & 0.838   \\
\texttt{Book-Crossing}                        & 3.987&5.095  & \bf{3.578}  & 3.704  & 3.804  & 4.185  \\
\texttt{ML10M}                    & 0.706&0.777  & \bf{0.655}  & 0.656  & 0.656  & 0.656    \\
\texttt{Jester}                     & 0.214 &0.217 & \bf{0.160}  & 0.162  & 0.162  & 0.165    \\
\texttt{FilmTrust}                  & 0.820  &0.941& \bf{0.668}  & 0.735  & 0.711  & 0.742  \\\bottomrule
\end{tabular}
    \end{scriptsize}
\end{table}

Experimental RMSEs of different sampling methods  are shown in Table.\;\ref{tab:datasets}.
%, where the best RMSE value of each dataset is marked in boldface.
Table.\;\ref{tab:datasets} shows that IGCS outperforms random sampling and LSS in all datasets, which have various data size, density and rating level.
Moreover, when the warm start parameter $\zeta$ in IGCS becomes larger, RMSE of IGCS deteriorates only slightly, but still significantly outperforms the competitors for almost all datasets.

%% file: table4.tex
\begin{table}
    \caption{RMSE and sampling time for IGCS with different $\zeta$'s on \texttt{ML100K}, along with LSS as comparison.}
    \centering
    \label{tab:movielens_100k2}
    \begin{scriptsize}
    \begin{tabular}{cccccccc}
    \toprule
       & MC &LSS&eigs &$\zeta=1$&$\zeta=3$&$\zeta=5$&$\zeta=7$ \\\midrule
      \multirow{4}*{G1}&
      GRALS&0.962&0.931&
      \textbf{0.927}&0.935& 0.934&0.931\\
      ~&GC-MC&0.910&0.896&
      \textbf{0.889}&0.895&0.897&0.891
\\
      ~&NMC&0.891&0.907&
      \textbf{0.880}&0.888&0.889&0.886
\\
\cmidrule(lr){2-8}
 ~&Time ($10^3$s)&-&1.975& 
 1.104&0.503& 0.375& 0.320\\\midrule    
        \multirow{4}*{G2}&   
        GRALS&0.958&0.889&
    0.871  &   \textbf{0.870}  &  0.882 &  0.882\\
    ~&GC-MC&0.909&0.860&
    \textbf{0.839}&	0.840&	0.847	&	0.851\\
      ~&NMC&0.907&0.858&
      \textbf{0.840}&0.845&	0.843&	0.852
\\
      \cmidrule(lr){2-8}

 ~&Time ($10^3$s)&-&1.278&
 1.216&0.573& 0.441& 0.388\\\bottomrule
    \end{tabular}
    \label{IGCS_k}
    \end{scriptsize}
\end{table}

%% file: conclude.tex
Pre-selection of entries for matrix completion is an important but under-addressed problem. 
In this paper, we propose a graph sampling strategy for matrix completion based on recurrent Gershgorin disc shift. 
Specifically, assuming that the target matrix signal is smooth with respect to dual graphs, we can complete the matrix via partial observations by solving a system of linear equations.
To maximize the stability of the linear system, we select samples to maximize the smallest eigenvalue $\lambda_{\min}$ of the coefficient matrix, which is equivalent to minimize the upper-bound of the reconstructed error.  
We tackle the formulated sampling objective with a greedy scheme to select one sample at a time (equivalent to shifting one Gershorin disc).  
To achieve fast sampling, inspired by one corollary of the Gershgorin circle theorem, we select the node corresponding to the largest energy in the first eigenvector of the incremented Laplacian matrix. 
We employ LOBPCG to compute the first eigenvector of an incremented Laplacian matrix, which benefits from warm start as the first eigenvectors are computed repeatedly. 
To efficently sample large real-world datasets, we further devise a block-wise graph sampling scheme, where the samples are collected alternately between blocks in two separate block-diagonal matrices.  
Extensive experiments have validated the superiority of our proposed graph sampling method for matrix completion, compared with other graph sampling and active matrix completion methods, in different datasets, under different graphs, and in combination with different popular matrix completion methods.

%% file: appen.tex
\section{Derivations of linear equation}
\label{solution}
We first compute the derivative of $f(\X)$ with respect to the optimization variable $\mathbf{X}$:
\begin{equation}
\frac{\partial f(\mathbf{X})}{\partial\mathbf{X}}=\mathbf{A}_{\Omega}\circ(\mathbf{X}-\mathbf{Y})
+\alpha\mathbf{L}_{r}\mathbf{X}
+\beta\mathbf{X}\mathbf{L}_{c}
\end{equation}
whose vector form by using the property $\text{vec}(\A\B)=(\I_n\otimes\A)\text{vec}(\B)=(\B^{\top}\otimes\I_m)\text{vec}(\A)$ is: 
\begin{align}
&
\frac{\partial f[\text{vec}(\mathbf{X})]}{\partial[\text{vec}(\mathbf{X})]}
\\
&=
\left(\tilde{\mathbf{A}}_{\Omega}+\alpha\mathbf{I}_n\otimes\mathbf{L}_{r}
+\beta\mathbf{L}_{c}\otimes\mathbf{I}_m\right)\text{vec}(\mathbf{X})
-\text{vec}(\mathbf{Y})\nonumber
\end{align}
Note that $\mathbf{A}_{\Omega}\circ\mathbf{A}_{\Omega}=\mathbf{A}_{\Omega}$ and $\Y=\A_{\Omega}\circ(\X+\N)$, so $\text{vec}(\mathbf{A}_{\Omega}\circ\mathbf{Y})=\text{vec}(\Y)$.  
To obtain an optimal solution, we set ${\partial f(\mathbf{X})}/{\partial\mathbf{X}}=\0$, which in vector form leads to 
\begin{equation}
    \left(\tilde{\mathbf{A}}_{\Omega}+\alpha\mathbf{I}_n\otimes\mathbf{L}_{r}
+\beta\mathbf{L}_{c}\otimes\mathbf{I}_m\right)\text{vec}(\mathbf{X}^{*})=\text{vec}(\mathbf{Y})
\label{eq:app-linear}
\end{equation}

\section{The property of matrix $\Q$}
\label{Q}
 Note that $\lambda_{\min}(\mathbf{I}_n\otimes\mathbf{L}_{r})=0$ since $\lambda_{\min}(\L_r)=0$ and  $\lambda_{\min}(\mathbf{L}_{c}\otimes\mathbf{I}_m)=\lambda_{\min}(\mathbf{I}_m\otimes\mathbf{L}_{c})=0$, so $\lambda_{\min}$ of $\alpha\mathbf{I}_n\otimes\mathbf{L}_{r}
+\beta\mathbf{L}_{c}\otimes\mathbf{I}_m$ is at least 0 based on Weyl's inequality on eigenvalues that $\lambda_{\min}(\A+\B)\ge\lambda_{\min}(\A)+\lambda_{\min}(\B)$ \cite{neumannseries}.
Moreover, the vectorized sampling operator $\tilde{\A}_{\Omega}$ is positive semi-definite (PSD). 
If matrix $\Q$ is invertible with enough samples in matrix $\tilde{\A}_{\Omega}$, we know that $\Q=\tilde{\A}_{\Omega}+\alpha\mathbf{I}_n\otimes\mathbf{L}_{r}
+\beta\mathbf{L}_{c}\otimes\mathbf{I}_m$ is sparese, symmetric and positive definite (PD), and the optimal solution to problem \eqref{eq:app-linear} in closed form is :
\begin{equation}
        \text{vec}(\mathbf{X}^{*})=\left(\tilde{\mathbf{A}}_{\Omega}+\alpha\mathbf{I}_n\otimes\mathbf{L}_{r}
+\beta\mathbf{L}_{c}\otimes\mathbf{I}_m\right)^{-1}\text{vec}(\mathbf{Y})
\label{eq:app-solution}
\end{equation}

\section{The upper-bound of $\lambda_{\max}(\Q)$}
\label{upper-bound}
Reusing the notations in Section\;\ref{sec:sampling1}, let $\L=\alpha\I_n\otimes\L_r+\beta\L_c\otimes\I_m$, and $\Q=\tilde{\A}_{\Omega}+\L$. 
Specifically, 
\begin{align}
&   \L= \alpha
\left[ \begin{array}{cccc}
\L_r & 0 & \cdots & 0 \\
0 & \L_r &  & \vdots \\
\vdots & & \ddots &0\\
0 & \cdots & \cdots&\L_r
\end{array} \right]\\
&+\beta
\left[ \begin{array}{cccc}
 \mathbf{L}_c(1,1)\mathbf{I}_m & \mathbf{L}_c(1,2)\mathbf{I}_m & \cdots & \mathbf{L}_c(1,n) \mathbf{I}_m \\
 \mathbf{L}_c(2,1)\mathbf{I}_m & \mathbf{L}_c(2,2)\mathbf{I}_m & \cdots & \mathbf{L}_c(2,n) \mathbf{I}_m \\
 \vdots & \vdots & \ddots & \vdots \\
\mathbf{L}_c(n,1)\mathbf{I}_m & \mathbf{L}_c(n,2)\mathbf{I}_m & \cdots & \mathbf{L}_c(n,n)\mathbf{I}_m
 \end{array} \right].\nonumber
\end{align}

For $k=i+m\times(j-1)$ with $\forall i\in \cV_r$ and $ j\in \cV_c$, the $k$-th row of matrix $\L$ (denoted by $\s_k$) is a combination of the $i$-th row of $\L_r$ (denoted by $\v_i$) and the $j$-th row of $\L_c$ (denoted by $\t_j$). 
%\red{check if the above is correct.}
It is easy to see that $\s_k(k)=\alpha\v_i(i)+\beta\t_j(j)$ and $\sum^{mn}_{l=1;l\neq k}|\s_k(l)|=\alpha\sum^m_{l=1;l\neq i}|\v_i(l)|+\beta\sum^n_{l=1;l\neq j}|\t_j(l)|$. Since $\v_i(i)=\sum^m_{l=1;l\neq i}|\v_i(l)|$ and $\t_j(j)=\sum^n_{l=1;l\neq j}|\t_j(l)|$ by the definitions of combinatorial Laplacian matrix $\L_r$ and $\L_c$, we know that $\s_k(k)=\sum^{mn}_{l=1;l\neq k}|\s_k(l)|,\forall k\in\{1,2,\dots,mn\}$. Based on the Gershgorin circle theorem presented in Section \ref{gershgorin-sec}, we know that the eigenvalues of matrix $\L$ are all bounded in [0,$2\max_k\{\s_k(k)\}$], where 0 is the lower bound of all left ends of Gershgorin discs, and $2\max_k\{\s_k(k)\}$ is the upper bound of all right ends of those discs. 
Note that $\v_i(i)=\D_r(i,i)$ and  $\t_j(j)=\D_c(j,j)$ for connected graph without self-loop. 
Assuming $\max_i\{\D_r(i,i)\}\le d_r$ and $\max_j\{\D_c(j,j)\}\le d_c$ (degree constrained graphs), we will have 
\begin{align}
  &  \max_k\{\s_k(k)\}=\alpha\max_i\{\v_i(i)\}+\beta\max_j\{\t_j(j)\}\\
  &=\alpha\max_i\{\D_r(i,i)\}+\beta\max_j\{\D_c(j,j)\}\le \alpha d_r+\beta d_c\nonumber
\end{align}

Therefore, $\lambda_{\max}(\L)$ will be upper-bounded by $2\alpha d_r+2\beta d_c$. Because $\tilde{\A}_{\Omega}$ is a diagonal matrix with diagonal entries 0 or 1, $\lambda_{\max}(\Q)$ will be  upper-bounded by $2\alpha d_r+2\beta d_c+1$ if the two factor graphs are degree-bounded by $d_r$ and $d_c$ respectively.